\documentclass[11pt]{article}
\usepackage{amsmath,amssymb}
\usepackage{mathtools}
\usepackage{sn-preamble-v2}
\usepackage{soul}
\usepackage{hyperref}
\hypersetup{
    colorlinks=true,
    citecolor=blue,
    linkcolor=purple,
    filecolor=magenta,      
    urlcolor=cyan,
    breaklinks=true
    }
\usepackage[nameinlink]{cleveref}

\Crefname{fact}{Fact}{Facts}
\Crefname{claim}{Claim}{Claims}
\usepackage[margin=1in]{geometry}
\usepackage{thmtools} 
\usepackage{thm-restate}
\usepackage[dvipsnames]{xcolor}
\usepackage{algorithm, algorithmic}
\usepackage{enumitem}
\usepackage[inline]{asymptote}
\usepackage{caption}
\usepackage[title]{appendix}
\usepackage{quiver}
\usepackage[export]{adjustbox}
\usepackage{booktabs}
\usepackage{multirow}
\usepackage{cite}

\title{On the Complexity of Mixed Equilibria \\ in First-Price Auctions with Correlated Priors}
\author{CCCPS}

\author{
Mark Chen
\\
cmark@cs.stanford.edu
\\
Stanford University
\and
Xi Chen
\\
xc2198@cs.columbia.edu
\\
Columbia University
\and
Hao Cui
\\
hc3377@columbia.edu
\\
Columbia University
\and
William Pires
\\
wp2294@columbia.edu
\\
Columbia University
\and
Jonah Stockwell
\\
js5384@columbia.edu
\\
Columbia University
}

\newcommand{\PPAD}{\textsf{PPAD}}
\newcommand{\NP}{\textsf{NP}}
\newcommand{\FIXP}{\textsf{FIXP}}
\newcommand{\Gcircuit}{\textnormal{\textsc{GCircuit}}}

\newcommand{\xx}[1]{{\color{blue}x_{#1}}}
\newcommand{\yy}[1]{{\color{red}y_{#1}}}

\newcommand{\x}{{\color{blue}x}}
\newcommand{\y}{{\color{red}y}}

\newcommand{\h}{{\color{OliveGreen}h}}
\newcommand{\hh}[1]{{\color{OliveGreen}h_{#1}}}
\newcommand{\z}[1]{{\color{Sepia} z_{#1}}}

\newcommand{\p}[1]{p({#1})}

\begin{document}

\maketitle

\begin{abstract}
We show that computing an approximate mixed Bayes–Nash equilibrium in a discrete first-price auction with correlated priors is \PPAD-complete. The key intermediate step in our reduc\-tion is the \PPAD-completeness of computing an approximate Nash equilibrium in a new normal-form game, the \emph{hypergraph discrete first-price auction}, which may be of independent interest.

\end{abstract}
\newpage 
{\small \tableofcontents}
\newpage

\section{Introduction}
One of the most fundamental and widely adopted auction formats is the \emph{first-price} auction, in which the highest bidder wins the item and pays her bid. In particular, first-price auctions and their variants are extensively used in online display advertising \cite{Renato20,Stylianos21,Conitzer22, Aggarwal24}.
However, unlike the second-price auction, 
  in which bidding truthfully is a dominant strategy~\cite{Vickrey61}, the first-price auction incentivizes bidding strategically based on one's belief about the other bidders' values. The game-theoretic study of first-price auctions dates back to the seminal work~of \cite{Vickrey61} and has remained a central theme of auction theory.
The existence, uniqueness, and characterization of their Bayes-Nash equilibria, within
the framework of Bayesian games of incomplete
information~\cite{Harsanyi67}, have been studied
extensively
\cite{Griesmer67,Riley71,Chwe89,Maskin85,Lebrun99,Maskin00,Maskin03,Plum92,Marshall94,Lebrun96,Lizzeri00,Athey01,Reny04,Lebrun06,Athey07,Chawla13, Bergemann17}. 

In recent years, a line of work in theoretical computer science has studied the computational aspect of these auctions: \emph{how hard is it to compute Bayes-Nash equilibria of first-price auctions?} A number of hardness results have been obtained for various
settings, which we review in \Cref{sec:related}.
However, the complexity of the most
fundamental setting, where bidders have independent private values (IPV),  
remains wide open~\cite{FGHLP23,filosratsikas2024computationequilibriadiscretefirstprice}. A natural intermediate step, as suggested
by~\cite{filosratsikas2025equilibriumcomputationfirstpriceauctions}, is the setting with \emph{correlated priors}, where the values of bidders are
drawn from a common joint distribution; this is the setting we study.

Concretely, we consider  \emph{discrete first-price auctions} (DFPA)
\emph{with correlated priors}:
\begin{flushleft}\begin{enumerate}
    \item[] There are $n$ bidders.
    Let $B$ be a finite subset of $[0,1]$ that denotes the bidding space, i.e., the set of bids they can submit.
    We are also given a joint probability distribution $\calD$ over $[0,1]^n$ with finite support, where each point specifies the values of the $n$ bidders. 
    Both $B$ and $\calD$ are given explicitly in the input; in particular, $\calD$ is given as a sequence of points in $[0,1]^n$ and their probability masses.
    The goal is to find an approximate mixed Bayes-Nash equilibrium (MBNE).
    Roughly speaking, an equilibrium consists of a mixed strategy $\beta_i$ for each bidder $i$, which maps each value in $V_i$ (the support of bidder $i$'s values in $\calD$) to a probability distribution over $B$ such that bidding according to $\beta_i$ is an approximate best response (i.e., to approximately maximize her expected utility) with respect to other bidders' strategies when their values are drawn from $\calD$.
    (See \Cref{sec:defDFPA} for the definition.)
\end{enumerate}\end{flushleft}
In \cite{filosratsikas2025equilibriumcomputationfirstpriceauctions}, the problem of deciding the existence of a pure
Bayes-Nash~equilibrium is shown to be \NP-hard. Mixed Bayes-Nash equilibria, in contrast, always
exist in this setting \cite{filosratsikas2025equilibriumcomputationfirstpriceauctions}, so computing one is a total search problem; its complexity was left open
by~\cite{filosratsikas2025equilibriumcomputationfirstpriceauctions}.

\subsection{Our Results}
We 
  resolve the complexity of computing an MBNE of a DFPA with correlated priors:

\begin{theorem} \label{thm:DFPA_ppad_hard}
    There exists a constant $\epsilon>0$ such that the problem of computing an $\epsilon$-approximate MBNE in a DFPA with correlated priors is $\PPAD$-complete. 
{The \PPAD-hardness holds even when the bidding space $B$ is $\{0,1/5,2/5\}$.}
\end{theorem}

En route towards proving \Cref{thm:DFPA_ppad_hard}, a crucial step is the introduction of a new game which we call \emph{Hypergraph Discrete First-Price Auctions} (HDFPA):
\begin{flushleft}\begin{enumerate}
    \item[] As in the DFPA, there are $n$ bidders and $B$ is a finite subset of $[0,1]$ that denotes the bidding space. 
    \emph{Unlike in a DFPA with correlated priors, each bidder in an HDFPA has a fixed value $v_i\in [0,1]$.} 
    The auction that (a subset of) bidders will participate in, on the other hand, is drawn from a distribution. 
    Formally, we are given a weighted hypergraph $H$ explicitly, which   defines a distribution over subsets of $[n]$.
    The game proceeds as follows: First each bidder submits a bid in $B$; then a subset $\bS\subseteq [n]$ is drawn from the weighted hypergraph $H$; finally, a
    first-price auction is held between bidders in $\bS$ using their bids submitted at the beginning.
    A Nash equilibrium (NE)\footnote{As it becomes clear later, an HDFPA is a standard normal-form game.} consists of a distribution over $B$ for each bidder as a best response mixed strategy to maximize her expected utility with respect to other bidders' mixed strategies.
    (See \Cref{sec:defHDFPA} for the formal definition, especially the definition of approximate NE in \Cref{def:normalized-ne}, which is slightly different from the standard definition to simplify the presentation.)
\end{enumerate}\end{flushleft}

The new model of HDFPA, in our opinion, is considerably easier to work with compared to DFPA. 
It was inspired by the work of \cite{filosratsikas2025equilibriumcomputationfirstpriceauctions} on the NP-hardness of
deciding the existence of pure BNE in DFPA. %
Examining their construction, we note that in every point in $\supp(\calD)$, all but at most four bidders have value $0$.
Intuitively, one can view each point in $\supp(\calD)$ as an auction held between a small set of at most four bidders with positive values. 
Now this is not formally correct because even bidders with value $0$ will submit bids of $0$ (given the no-overbidding rule) and still participate in the auction, which adds complexity to the design and analysis of the construction of \cite{filosratsikas2025equilibriumcomputationfirstpriceauctions}. 
In contrast, in the HDFPA, bidders outside of the drawn set $\bS$ do not participate in the auction at all, simplifying the construction and its analysis.

A more noticeable simplification the HDFPA model provides is that it is a standard normal-form game with complete information. Each bidder $i$ in an HDFPA has a fixed value $v_i\in [0,1]$ known to all bidders. Therefore, we work with Nash equilibria instead of Bayes-Nash equilibria, and do not have to consider conditional distributions such as $\calD_{i|v}$ (i.e., distributions of other bidders' values conditioning on bidder $i$ having value $v$), which is another complication in the work of  \cite{filosratsikas2025equilibriumcomputationfirstpriceauctions}.

Despite its simplicity, 
  HDFPA serves as a key bridge to establish \Cref{thm:DFPA_ppad_hard}.
First we show that the problem of finding an approximate Nash equilibrium in an HDFPA is \PPAD-complete:

\begin{theorem} \label{thm:HDFPA_ppad_hard}
    There exists a constant $\epsilon >0$ such that the problem of computing an $\epsilon$-approximate NE in an HDFPA is \PPAD-complete. 
{The \PPAD-hardness holds even when the bidding space $B$ is $\{0,1/4\}$ and each auction in the distribution has no more than four participants.} 
\end{theorem}

Next we prove \Cref{thm:DFPA_ppad_hard} by 
  giving a polynomial-time reduction from hard 
  HDFPA instances constructed in the proof of \Cref{thm:HDFPA_ppad_hard}
  to DFPA with correlated priors  in \Cref{sec:maintheorem}.
We give an overview of the proofs of both \Cref{thm:HDFPA_ppad_hard} and \Cref{thm:DFPA_ppad_hard} in \Cref{sec:techoverview}.

To complete the picture, we also show that the problem of deciding whether an HDFPA has a pure NE or not is NP-complete. This is proved in \Cref{sec:npc}.

\subsection{Related Work}\label{sec:related}

In \cite{FGHLP23}, the authors studied the first-price auction with continuous value distributions and a discrete bidding space. They showed that the problem of computing an $\eps$-approximate pure BNE is \PPAD-complete and that computing an exact one is \FIXP-complete. Their results hold~when~the priors are \emph{subjective}, meaning that each bidder $i$ holds her own belief $F_{i,j}$ about the value of~each other bidder $j$, and these beliefs are independent across $j$. \cite{ChenPeng2023} subsequently established \PPAD-hardness for the same setting with common priors. However, their construction  relies on a non-standard (``trilateral'') tie-breaking rule that is part of the input. Under the standard uniform~tie-breaking rule, on the other hand, they gave a polynomial-time approximation scheme (PTAS).

Recently \cite{filosratsikas2024computationequilibriadiscretefirstprice} turned to the \emph{discrete} first-price auction, in which both values and bids are discrete.  
For subjective priors, 
\cite{filosratsikas2024computationequilibriadiscretefirstprice} showed that the problem of deciding the existence of a pure BNE is \NP-complete and that of computing an MBNE, which always exists, is \PPAD-complete.
Subsequently, \cite{filosratsikas2025equilibriumcomputationfirstpriceauctions} studied the discrete first-price auction with common correlated priors. They proved that deciding the existence of a pure BNE remains \NP-hard, and left
  the complexity of MBNE in this setting as an open problem. 
This is the problem we settle in \Cref{thm:DFPA_ppad_hard}.

\section{Technical Overview}\label{sec:techoverview}

We give overviews of the 
  \PPAD-hardness proofs of HDFPA and DFPA.
For
convenience, the proof for HDFPAs (\Cref{sec:mainreduction}) uses values in $[0,4]$ and
bids in $[0,1]$, while the proof for DFPAs (\Cref{sec:maintheorem}) uses values in $[0,5]$ and
bids in $[0,2]$; we use the same ranges in overviews below. These choices~are without loss
of generality: dividing all values and bids by $4$ (in HDFPAs) or by $5$ (in DFPAs) 
brings them back into $[0,1]$. This rescales utilities of bidders by the same factor, so an
$\eps$-approximate NE or MBNE of the original instance is an
$(\eps/4)$- or $(\eps/5)$-approximate NE or MBNE of the normalized one. The same hardness results hold after scaling down the approximation parameter accordingly.

\subsection{\PPAD-hardness of HDFPA}%

To prove the \PPAD-hardness of HDFPA, we 
  give a polynomial-time reduction from the 
  $\Gcircuit$ problem \cite{BudishEtAl2023} with 
  SUM and NOT gates,
  which is known to be \PPAD-hard to approximate for some sufficiently small
  constant $\eps>0$. 
The input of the problem is a \emph{generalized circuit} $C=(V, \calT)$, where each node $u\in V$ is associated with a variable $\alpha_u$, and $\calT$ is a collection of SUM and NOT gates applied on nodes in $V$. The goal is to find an assignment $\alpha=(\alpha_u)\in [0,1+\eps]^V$ of the variables that approximately satisfies every gate constraint. 
For example, a SUM gate with inputs $u,v$ and output $w$ requires that
$\alpha_w=\min(\alpha_u+\alpha_v,1) \pm \epsilon.$
See \Cref{sec:Gcircuit} for the formal definition. Recall that each node can be the output of at most one gate in $\calT$.
Moreover, the construction of \cite{BudishEtAl2023} ensures that each node is an input to at most two gates. 
\medskip

\noindent \textbf{The gadget paradigm of \cite{GDP09}:} To build intuition for our proof, we take a conceptual look at how \cite{GDP09}  uses \emph{game gadgets} to implement generalized circuits in graphical games.

Given a generalized circuit $C=(V,\calT)$, 
  the construction of \cite{GDP09} introduces, for each node $u\in V$, a player $y_u$ with action set $\{0,1\}$ in the graphical game $G$.
Given a mixed strategy profile $p$ of $G$, let {$\p{y_{u}}\in [0,1]$} denote the probability that $y_u$ plays action $1$.
The goal is to design the~payoffs of $G$ so that,
at every $\epsilon'$-approximate NE $p$ of $G$, for a sufficiently small constant $\eps'>0$,
  setting $\alpha_u=p(y_u)$ for all $u\in V$ 
  approximately satisfies every gate constraint in $\calT$. 

This is achieved by \emph{game gadgets}. 
For example, a SUM gate with input~nodes $u,v$ and output node $w$ is simulated by a {game gadget} that involves only the players $y_{u}, y_{v},y_{w}$ and a few extra helper players created specifically for that gadget.
Crucially, the game gadget is  \emph{directed}: 
it specifies the payoffs of its output player $y_w$ and its own helper players only, and leaves 
  the payoffs of its input players $y_u$ and $y_v$ unspecified.
It guarantees that, 
  regardless of how $y_u$ and $y_v$ play  (which depends on how their payoffs are defined),
  as long as $y_w$ and 
  the helper players play $\eps'$-best responses, the desired constraint always holds:
\begin{equation}\label{eq:haha1}
\p{y_{w}}=\min\big(\p{y_{u}}+\p{y_{v}},1\big) \pm \epsilon
\end{equation}
All other game gadgets of \cite{GDP09} are directed in the same way. 
Since each node of a generalized circuit is the output of at most one gate, and each helper player belongs to a single game gadget, every player's payoff is specified by at most one game gadget. The game gadgets therefore do~not interfere with one another, and each
enforces its own constraint, such as~\eqref{eq:haha1}. Consequently, every $\eps'$-approximate NE of $G$ yields an approximate solution to the generalized circuit $C$.\medskip

\noindent \textbf{The  main challenge: lack of directedness.} 
At a high level, we would like to follow the gadget paradigm of \cite{GDP09} and simulate generalized circuits by HDFPAs, using   \emph{auction gadgets} in place of game gadgets.
To this end, we introduce~a bidder $x_u$ for each node $u\in V$ of the generalized circuit $C=(V,\calT)$, and set
 {the bidding space $B$ to be $\{0,1\}$}. As before, for a mixed strategy profile $p$ of the HDFPA $H$ that we construct, we write $p(x_u)$ for the probability that $x_u$ bids $1$.
The goal is to design the distribution of auctions (i.e., the weighted hypergraph over bidders) so that, at every $\eps'$-approximate NE $p$ of $H$, setting $\alpha_u=p(x_u)$ approximately satisfies every gate constraint in $\calT$.

For example, for a SUM gate in $\calT$ with input nodes $u,v$ and output node $w$, we would like
to design an auction gadget, i.e., a weighted collection of auctions among $x_u$, $x_v$, $x_w$ and
possibly some helper bidders created specifically for the gadget, such
that~\eqref{eq:haha1} (after replacing $y$'s with $x$'s) holds for  at every approximate NE. Designing such an auction gadget in
isolation is not difficult. The challenge arises when gadgets are combined to work together. A game
gadget can be made directed simply by leaving the payoffs of its input players
unspecified. An auction, by contrast, affects \emph{all} of its participants: if two bidders $x$
and $y$ take part in an auction, then $x$'s utility depends on how $y$ bids,
\emph{and vice versa}. Since a bidder $x_u$ may take part in several gadgets (as the
output of one gate and as an input of at most two gates), it is unclear how to prevent the auction gadgets
from interfering with one another, even if each of them enforces its own constraint in
isolation.\medskip

\noindent\textbf{Our construction.}
For each node $u\in V$ of the generalized circuit $C=(V,\calT)$, we introduce \emph{two} bidders in
the HDFPA $H$: an \emph{\color{red}output bidder} $\yy{u}$ with value~$4$ and an \emph{\color{blue}input
bidder} $\xx{u}$ with value~$2$.\footnote{Throughout the paper, we use color coding
to indicate the role of each bidder: {\color{blue}blue} for input bidders and
{\color{red}red} for output bidders. }
Intuitively, $\yy{u}$ represents $u$ as the output of its gate, while $\xx{u}$
represents $u$ whenever it is used as an input to other gates. The two bidders are
linked by a \emph{COPY auction gadget} from $\yy{u}$ to $\xx{u}$, which has two
properties:
\begin{flushleft}\begin{enumerate}
\item \emph{Copying:} regardless of how $\yy{u}$ plays, as long as $\xx{u}$ and helpers of the gadget play
  $\eps'$-best responses, we have $\p{\xx{u}} = \p{\yy{u}} \pm \eps$: the
  strategy of $\yy{u}$ is approximately copied to $\xx{u}$.
\item \emph{Directedness:} under \emph{every} mixed strategy profile, the auctions of the COPY gadget
  are\\ such that there is no difference in the payoff for $\yy{u}$ when she bids $0$ or $1$. I.e. $\yy{u}$'s optimal strategy isn't impacted by the gadget at all. 
\end{enumerate}\end{flushleft}
We present the COPY gadget in \Cref{sec:copygadget}. 

With the COPY gadget in hand, we build the HDFPA $H$ as follows:
\begin{flushleft}\begin{enumerate}
\item For each node $u\in V$, we add a COPY gadget from $\yy{u}$ to $\xx{u}$.
\item For each SUM gate in $\calT$ with input nodes $u,v$ and output node $w$, we add a
  SUM gadget over $\xx{u}$, $\xx{v}$, $\yy{w}$ and a few helper bidders such that,
  regardless of how $\xx{u}$ and $\xx{v}$ play, as long as $\yy{w}$ and the helpers
  play $\eps'$-best responses with respect to the gadget, we always have
  \begin{equation}\label{eq:sum-auction}
    \p{\yy{w}} = \min\bigl(\p{\xx{u}}+\p{\xx{v}},\,1\bigr) \pm \eps.
  \end{equation}
\item For each NOT gate in $\calT$ with input node $u$ and output node $w$, we add a
  NOT gadget over $\xx{u}$, $\yy{w}$ and a few helper bidders such that, regardless of
  how $\xx{u}$ plays, as long as $\yy{w}$ and the helpers play $\eps'$-best responses
  with respect to the gadget, we always have
  \begin{equation}\label{eq:not-auction}
    \p{\yy{w}} = 1-\p{\xx{u}} \pm \eps.
  \end{equation}
\end{enumerate}\end{flushleft}
Here, ``with respect to the gadget'' means that only the utility from the gadget's own
auctions is taken into account. Unlike the COPY gadget, the SUM and NOT gadgets are
\emph{not} directed.

To see why these gadgets can be combined, consider the two types of bidders
separately.
\begin{flushleft}\begin{enumerate}
\item \emph{Output bidders.} Each $\yy{w}$ takes part in at most two gadgets: the COPY
  gadget from $\yy{w}$ to $\xx{w}$, and the SUM or NOT gadget whose output node is~$w$.
  Since the COPY gadget is directed, all of $\yy{w}$'s gains from switching bids come from the latter.
  Hence, at every $\eps'$-approximate NE of $H$, $\yy{w}$ plays an $\eps'$-best
  response with respect to that gadget, and \eqref{eq:sum-auction}
  or~\eqref{eq:not-auction} holds.
\item \emph{Input bidders.} Each $\xx{u}$ takes part in the COPY gadget from $\yy{u}$
  to $\xx{u}$, and in at most two SUM or NOT gadgets as an input. Since the SUM and NOT
  gadgets are not directed, they do affect $\xx{u}$'s utility. To control this effect,
  we give the auctions of the COPY gadget weights larger than those of the SUM and NOT
  gadgets by a sufficiently large factor.
  The SUM and NOT gadgets then have only a small influence on $\xx{u}$, so $\xx{u}$
  plays mainly with respect to the   COPY gadget, and
  $\p{\xx{u}} = \p{\yy{u}} \pm O(\eps)$ is  still enforced.
\end{enumerate}\end{flushleft}
Consequently, at every $\eps'$-approximate NE of $H$, setting $\alpha_u = \p{\yy{u}}$
for all $u\in V$ approximately satisfies every gate constraint of the generalized circuit $C$.

\subsection{\PPAD-hardness of DFPA}

To this end, we would like to reduce the problem of computing an $\epsilon$-approximate NE of an HDFPA to that of computing an $\epsilon'$-approximate MBNE of a DFPA with correlated priors, and 
it suffices~to focus on HDFPAs constructed in the proof of \Cref{thm:HDFPA_ppad_hard}. In particular, we assume that the bidding space $B=\{0,1\}$ and every auction in the hypergraph has at most four participants. 

Let $H$ be an HDFPA with $n$ bidders $[n]$, and let $v_i$ be the value of bidder $i$.
Let $E$ be the set of weighted hyperedges, where a hyperedge $e=(S,w_e)\in E$ means that with probability $w_e$, bidders in $S$ are selected to participate in the first-price auction.
A first natural attempt is the following:
\begin{flushleft}\begin{itemize}
    \item[]\textbf{Attempt \#1}: We build a DFPA using the same bidders $[n]$ and the same bidding space $B=\{0,1\}$. For each hyperedge $e = (S, w_e)$ $\in E$, add a point $p$ to the support of the correlated value distribution $\calD$ in DFPA as follows: for each $i \in S$, set the value of bidder $i$ in $p$ to be $p_i=v_i$, its value in the HDFPA, and for each $i \not \in S$ set $p_i=0$. Set the probability mass of this point $p$ to be the same $w_e$.
    
    \item[]Given an approximate MBNE of the above DFPA, for each bidder $i$,  we take her mixed strategy when her value is $v_i$  to use it as her mixed strategy  in the HDFPA, and hope that this gives an approximate NE of the original HDFPA.
\end{itemize}
\end{flushleft}

This construction, however, does not quite simulate the original auctions. In the HDFPA, only bidders in $S$ participate in the auction when $e=(S,w_e)$ is selected and thus, bidders $i\notin S$ can never win the item in this situation. In the DFPA, even when values of bidders drawn from $\calD$~are  $p$, all bidders are present, including those $i\notin S$ with $p_i=0$ who will always bid $0$ according to the no-overbidding rule. These bids of $0$ cannot simply be ignored: when every bidder $i \in S$ bids $0$, all bidders will split the item following the uniform tie-breaking rule.
As a result, the probability of a bidder $i\not\in S$ winning the item can be positive in general.

\begin{flushleft}\begin{itemize}
    \item[]\textbf{Attempt \#2.}
To bypass the subtle issue above, we ``\emph{shift}'' both values and bids in the construction as follows. In the DFPA, the new bidding space is $B^+=\{0,b_0^+=1,b_1^+=2\},$ 
  where $b_0^+$ and $b_1^+$ correspond to the $0$ and $1$ in the original bidding space $B=\{0,1\}$ of the HDFPA, after being shifted up by $1$.
For each hyperedge $e = (S, w_e)\in E$, we similarly add a point $p$ with the same mass $w_e$ to $\calD$ in DFPA but $p$ is also shifted up by $1$ for bidders $i\in S$: %
 every bidder $i\in S$ 
 has value $p_i=v_i':=v_i+1$ and every bidder $i\notin S$ still has $p_i=0$ in $p$.

\item[]
Given an approximate MBNE of the above DFPA, for each bidder $i$, we \emph{hope} that when $i$ gets value $v_i'$, her mixed strategy is supported on $\{b_0^+,b_1^+\}$ only. 
Assuming this is the case, for each bidder $i$, we take her mixed strategy when her value is $v_i'$ and use it as her mixed strategy in the HDFPA after shifting down the bids by $1$ (so $b_0^+$ and $b_1^+$ correspond to $0$\\ and $1$, respectively).
We hope that this gives an approximate NE of the original HDFPA.
\end{itemize}\end{flushleft}

 Suppose for a moment that our wish comes true, 
 that any bidder $i$ in an approximate MBNE of the DFPA can only bid $\{b_0^+,b_1^+\}$ when her value is $v_i'$.
 We show that in this case, auctions in  the HDFPA are simulated perfectly in the DFPA, and the reduction would go through smoothly.
To see this, note that when $p$ is drawn in the DFPA, bidders $i\notin S$ have value $0$ and will always bid $0$. Given that all bidders $i\in S$ only bid on $\{b_0^+,b_1^+\}$, bidders $i\notin S$ can never win the item. Hence, we can ignore them from the auction and thus, only bidders in $S$ participate in the auction, just like in the HDFPA when $S$ is drawn. Furthermore, since we shifted both the bids and values up by $1$ in the DFPA, the utility of each bidder $i\in S$  is also preserved. %

The challenge of course is that we cannot simply force a bidder to bid $\{b_0^+,b_1^+\}$ only when her value is $v_i'$.
Nonetheless, it turns out that we can make some small changes to the construction in \textbf{Attempt \#2} to ensure that, at any approximate MBNE of the DFPA, every bidder $i$ only bids $0$ with a small probability when her value is $v_i'$.
This turns out to be good enough for the reduction to go through. 
The final construction is presented in \Cref{sec:finalconst}.

\section{Preliminaries}

\noindent\textbf{Notation.} Given a positive integer $n \geq 1$, we let $[n]$ denote the set $\{1, \ldots, n\}$. For a distribution $\calD$, we write $\supp(\calD)$ to denote its support and $\calD(v)$ to denote the probability mass of  $v\in \supp(\calD)$. For a finite set $X$, we denote by $\Delta(X)$ the set of all probability distributions over $X$. 

Given a tuple ${a}=(a_1, \ldots, a_k)$ and $i \in [k]$, we write $a_{-i}:=(a_1, \ldots, a_{i-1}, a_{i+1}{,} \ldots, a_k).$

\subsection{Discrete First-Price Auction with Correlated Priors}\label{sec:defDFPA}

A discrete first-price auction (DFPA) with correlated priors is given by a tuple $A=(\calN, B, {\calV}, \calD)$, where $\calN=[n]$ is the set of bidders, $B=\{b_0, \ldots, b_r\} \subseteq [0,1]$ with $0=b_0< \cdots< b_r \leq 1$ is their shared bidding space, and $\calD$ is a joint discrete probability distribution over $[0,1]^n$. The bidders' values are drawn according to  $\calD$, which is public and is known to every bidder: as $v\sim \calD$, bidder $i$'s value is set to $v_i$, which is private and only bidder $i$ knows. 
Finally, ${\calV}= (V_1, \ldots, V_n)$, where $V_i$ is the set of
  possible values of bidder $i$ in $\supp(\calD)$. We write $\calD_i$ to denote the marginal distribution of $\calD$ on bidder $i$'s value. So $V_i=\supp(\calD_i)$.
In an input instance $A=(\calN,B,\calV,\calD)$ of the problem, 
the bidding space $B$ is given explicitly by listing all bids $b\in B$; the distribution $\calD$ is also given explicitly as a list of points in $[0,1]^n$ along with their probabilities. 
All numerical parameters are rational and encoded in binary.

The game (with incomplete information) proceeds as follows.
First, a value vector $v$ is drawn from $\calD$ and each bidder $i$ receives its (private) value $v_i$.
Given value $v_i$, bidder $i$ picks a bid $b_i\in B$ according to a mixed strategy $\beta_i: V_i\rightarrow \Delta(B)$ by drawing $b_i\sim \beta_i(v_i)$. (Here $\beta_i$ is a pure strategy if $\beta_i:V_i\rightarrow B$, in which case bidder $i$ bids $\beta_i(v_i)$ without flipping coins.)

Once all bids $b_i$ are submitted, 
  a first-price auction on a single item is held.
The bidder who submitted the highest bid wins the item and pays her bid. When there is a tie for the highest bid, the winner is chosen according to {a} \emph{uniform tie-breaking rule}.  In more {detail}, given ${b} \in B^n$, we let 
$$M(b)=\max_{i\in [n]}b_i\quad\text{and}\quad h(b)=\big|\{i\in [n]:b_i=M(b)\}\big|{.}$$ For each bidder $i\in [n]$, her utility is then defined as
\[
u_i({b;v_i})=
\begin{cases}
 (v_i-b_i)\big/h(b) &\text{if } b_i=M(b)\\[0.2ex]
0 & \text{otherwise}
\end{cases}{.}
\]

\noindent\textbf{No overbidding:}
We assume that all bidders follow the no overbidding rule, i.e., a bidder never submits a bid that is higher than her value. 
Formally this means that 
  $\beta_i(v_i)$ can only be supported on $\{b\in B: b\le v_i\}$. 
This {is} a common assumption used in the literature   \cite{MASKIN2003395,LEBRUN2006131,ChenPeng2023,filosratsikas2024computationequilibriadiscretefirstprice}, which rules out unnatural equilibria. Indeed, submitting a bid higher than your value always {leads} to non-positive {utility, and} is thus weakly dominated by bidding $0$.
\medskip

\noindent
\textbf{Expected utility:} 
Given a mixed strategy profile $\beta=(\beta_1,\ldots,\beta_n)$, 
  consider bidder $i$ and one of her values $v_i\in V_i$. 
Her expected utility given value $v_i$ when bidding $b\in B$ is
$$U_{i}(b,\beta_{-i};v_i):=\E\limits_{{\mathbf{v}}_{-i} \sim \calD_{i|v_i} } \E_{{\mathbf{b}}_{-i} \sim \beta_{-i}({\mathbf{v}}_{-i})}\big[u_i(b,{\mathbf{b}}_{-i};v_i)\big],$$
where $\calD_{i|v_i}$ is the distribution of $v_{-i}$ from $v\sim \calD$ conditioning on the $i$-th coordinate being $v_i$. The expression above can be understood by the following process: First, we draw the values  ${\mathbf{v}}_{-i}$ for the other bidders (conditioned on bidder $i$ having value $v_i$). Next, we apply the other {bidders'} bidding strategy profile $\beta_{-i}$ on   ${\mathbf{v}}_{-i}$ to get their bids ${\mathbf{b}}_{-i}$. Finally we compute the utility of bidder $i$ given the bid profile $(b,{\mathbf{b}}_{-i})$ (and that she has value $v_i$).

Now, if bidder $i$ uses a mixed strategy $\beta_i:V_i \to \Delta(B)$, then her expected utility given value $v_i$ and strategy $\beta_i$ (with respect to $\beta_{-i}$) is as follows:
$$U_{i}(\beta_i,\beta_{-i}; v_i):=\E_{{\mathbf{b}} \sim \beta_i(v_i)}U_i({\mathbf{b}},\beta_{-i};v_i){.} $$

We are ready to define approximate Bayes-Nash equilibria of a DFPA with correlated priors:

\begin{definition}[Approximate  Mixed Bayes-Nash Equilibrium]\label{def:MBNEinDFPA}
    Let $\eps \geq 0$ and let $A=(\calN,B,\calV,\calD)$ be a DFPA with correlated priors. 
 A mixed strategy profile ${\beta}=(\beta_1, \ldots, \beta_n)$ is an $\eps$-approximate mixed Bayes-Nash equilibrium (MBNE) of $A$ if for any $i \in [n]$, $v_i\in V_i$ and  $b \in B$,
    $$U_{i}(\beta_i,\beta_{-i}; v_i) \geq U_{i}(b,\beta_{-i};v_i) -\eps{.} $$
    In other words, if bidder $i$ changes her strategy $\beta_i(v_i)$ to only bidding $b\in B$, then conditioned on her value being $v_i$, her expected utility could increase by at most $\eps$. 

    When $\beta_i$'s are  pure strategies, $\beta$ is an $\eps$-approximate pure Bayes-Nash equilibrium (PBNE).
\end{definition}

We are interested in the problem of computing an $\eps$-approximate MBNE of a DFPA with~correlated priors: given any $A=(\calN,B,\calV,\calD)$, return an $\eps$-approximate MBNE $\beta=(\beta_1,\ldots,\beta_n)$ of $A$, by listing explicitly probabilities in each distribution $\beta_i(v_i)$, $i\in [n]$ and $v_i\in V_i$.  \medskip

\noindent
\textbf{Normalization of bids and values:} In the proof of \Cref{thm:DFPA_ppad_hard} in \Cref{sec:maintheorem}, for convenience, we will use $B=\{0,1,2\}$ and a value distribution $\calD$ supported in $[0,5]^n$.
We show that the problem~of computing an $\eps$-approximate MBNE of such a game (under the same  \Cref{def:MBNEinDFPA} with additive error $\eps$) is \PPAD-hard for some constant $\eps>0$.
We can normalize by dividing both bids and values by $5$ so that the bidding space $B$ becomes $\{0,1/5,2/5\}$ (as in the statement of \Cref{thm:DFPA_ppad_hard}) and $\calD$ is supported in $[0,1]^n$.
It follows that the problem of computing an $(\eps/5)$-approximate MBNE, when bids and values are all within $[0,1]$, is \PPAD-hard.

\subsection{Hypergraph Discrete First-Price Auctions}\label{sec:defHDFPA}

A \emph{hypergraph discrete first-price auction} (HDFPA) is a normal-form game specified succinctly by~a tuple $H=(\calN,B,v,E)$ as follows.
First, $\calN=[n]$ is {the} set of bidders. Each bidder $i$ has a~value~$v_i$ $\in [0,1]$, with $v=(v_1,\ldots,v_n)$ which is public to all bidders.
The bidders are players in the~normal-form game.
Next, the bidding space $B \subseteq [0,1]$ is a finite set of nonnegative numbers, which always contains $0$.  
Each bidder $i\in [n]$ plays the game by picking a mixed strategy $\sigma_i\in \Delta(B)$.\medskip

\noindent\textbf{No overbidding:}
Similarly we will assume the \emph{no overbidding rule}: $\sigma_i$ is supported on bids in $B$ that are no larger than $v_i$ only. So bidder $i$ can 
only use a mixed strategy $\sigma_i\in \Delta(B_i)$, where $B_i:=$ $\{b\in B: b\le v_i\}$. As $0\in B$ and $v_i\ge 0$,   $B_i$ is always nonempty for every bidder $i\in [n]$.\medskip

Finally, we define utilities (or payoffs) of bidders in the game obtained from an action~profile $b\in B^n$, which are specified by the last component $E$. Here $E$ is a multiset of pairs $e=(S,w_e)$ where $S\subseteq [n]$ is a (non-empty) hyperedge and $w_e \in \mathbb{Q}_{>0}$ is its weight.
\emph{Different from the introduction}, we do \emph{not} require the $w_e$'s to sum to $1$; this is for the convenience of our construction and its analysis later in \Cref{sec:mainreduction}. 
We assume that every bidder belongs to at least one hyperedge in $E$.

We will refer to each $e\in E$ as a weighted hyperedge (auction) and $E$ as a weighted hypergraph (of auctions).
Each $e=(S,w_e)\in E$ should be viewed as a single-item first-price auction held~among bidders in $S$ only.
And as it becomes clear below, the utility of a bidder $i\in S$ obtained from this auction will be scaled multiplicatively by the weight $w_e$.
For each bidder $i\in [n]$, we write 
$$
W_i:=\sum_{(S,w_e)\in E:\hspace{0.05cm}i\in S} w_e
$$
to denote the total weight of hyperedges in $E$ that contain bidder $i$.
We have $W_i>0$ because it is assumed that every bidder belongs to at least one hyperedge in $E$.

Given an action profile $b\in B^n$ and a weighted hyperedge $e=(S,w_e)\in E$,
let 
$$M_e(b)=\max_{i\in S}b_i\quad\text{and}\quad h_e(b)=\big|\{i\in S:b_i=M_e(b)\}\big|.$$ 
Uniform tie-breaking gives bidder $i$ the following utility before being scaled by the weight $w_e$:
\[
u_i^e(b)=
\begin{cases}
 (v_i-b_i)\big/h_e(b) &\text{if\ }i\in S\ \text{and}\ b_i=M_e(b)\\[0.2ex]
0 & \text{otherwise}
\end{cases}{.}
\] %
Note that $u_i^e(b)$ is always $0$ if $i\notin S$ given that only bidders in $S$ participate in the auction.
The overall utility of bidder $i\in [n]$ under an action profile $b\in B^n$ is then defined as
$$
u_i(b)=\sum_{e=(S,w_e)\in E} w_e\cdot u_i^e(b).
$$
This finishes the description of the normal-form game specified by an HDFPA $H=(\calN,B,v,E)$.

Since this is a normal-form game, the existence of Nash equilibria follows from Nash's theorem \cite{Nash}.
Below we define $\eta$-NE of an HDFPA $H$. It is slightly different from the standard definition of~approximate NE, to simplify the presentation of \Cref{sec:maintheorem} where we reduce HDFPA to DFPA.

In more {detail}, given any mixed strategy profile $\sigma$, the expected utility of bidder $i$ is defined as
$$
U_i(\sigma):=\bE_{{\mathbf{b}}\sim \sigma} \big[u_i({\mathbf{b}})\big]
$$
and $\sigma$ (with $\sigma_i\in \Delta(B_i)$) is an  NE of $H$ if for every bidder $i\in [n]$ and any action $b_i\in B_i$, we have
$$
U_i(\sigma)\ge U_i(b_i,\sigma_{-i}),
$$
where $(b_i,\sigma_{-i})$ denotes the profile obtained from $\sigma$ by making bidder $i$ pick $b_i$ with probability $1$.

Allowing an additive error that is relative to $W_i$ for each bidder $i\in [n]$ (i.e., the total weight of hyperedges that contain $i$),
we define $\eta$-NE of an HDFPA as follows:

\begin{definition}[Approximate Nash Equilibria]\label{def:normalized-ne}
Let $\eta \geq 0$ and let $H=(\calN,B,v,E)$ be an HDFPA. A mixed strategy profile $\sigma=(\sigma_1, \ldots, \sigma_n)$, where $\sigma_i\in \Delta(B_i)$ for each $i$, is an $\eta$-NE of $H$ if for every bidder $i \in [n]$ and every $b_i\in B_i$, we have
\[
U_i(\sigma)\geq U_i(b_i,\sigma_{-i})-\eta W_i.
\]
We say $\sigma$ is an $\eta$-PNE (pure Nash equilibrium) of $H$ if $\sigma$ is a pure strategy profile.
\end{definition}

In the above, we consider an additive error of $\eta W_i$ because, as it becomes clear in \Cref{sec:maintheorem}, it allows us to connect better with MBNE of the DFPA we build there\footnote{There are other reasons why multiplying $\eta$ by $W_i$ is the ``right'' definition of approximate NE. For instance, if we multiply all the weights $w_e$ of an HDFPA by the same parameter $c>0$, it does not change the set of $\eta$-NE.}. In particular, we make the following observation:

 \begin{observation}[Alternative view {of} $\eta$-NE]\label{obs:alternative}
    Consider an HDFPA $H=(\calN,B,v,E)$ and a bidder~$i$; let $E_i$ be the distribution over edges $e=(S,w_e)$ with $i \in S$ (so $e$ is drawn with probability $w_e/W_i$). Then $\sigma$ is an $\eta$-NE of $H$ if and only if for every bidder $i\in [n]$ and $b_i \in B_i$, we have
     $$
\underset{\substack{\mathbf{e} \sim E_i, {\mathbf{b}} \sim \sigma \\}}{\E}\big[u^{\mathbf{e}}_i({\mathbf{b}})\big] \geq \underset{\substack{\mathbf{e} \sim E_i, {\mathbf{b}}_{-i} \sim \sigma_{-i}}}{\E}\big[u^{\mathbf{e}}_i(b_i, {\mathbf{b}}_{-i})\big]-\eta{.}$$
\end{observation}

In \Cref{thm:HDFPA_ppad_hard}, we are interested in the problem of computing an $\eta$-NE of an HDFPA. The input $H=(\calN,$ $B,v,E)$ is given as follows: numbers in $B$ and $v$ are given explicitly; $E$ is also~given~explicitly by listing all of its hyperedges $(S,w_e)$. 
The goal is to compute an $\eta$-NE $\sigma=(\sigma_1,\ldots,\sigma_n)$ of $H$, by listing explicitly probabilities in each distribution $\sigma_i\in \Delta(B_i)$.

\medskip
\noindent

\noindent
\textbf{Normalization of bids and values:} In the proof of \Cref{thm:HDFPA_ppad_hard} in \Cref{sec:mainreduction}, for convenience,~we will use bidding space $B=\{0,1\}$ but allow bidders to have $v_i\in [0,4]$.
We show that the problem~of computing an $\eta$-NE of such a game (under    \Cref{def:normalized-ne} with additive errors $\eta W_i$) is \PPAD-hard for some constant $\eta>0$.
We can normalize by dividing both bids and values~by $4$ so that $B$ becomes $\{0,1/4\}$ (as in \Cref{thm:HDFPA_ppad_hard}) and all $v_i\in [0,1]$.
It follows that the problem of computing an $(\eta/4)$-NE, when bids and values are in $[0,1]$, is \PPAD-hard.

\medskip\noindent
\textbf{A point of notation:} In some of our proofs, the bidding space $B$ is $\{0,1\}$. In this case we usually denote a mixed strategy profile by $p$ (rather than $\sigma$) and write  $\p{x}\in [0,1]$ {to denote} the probability of a bidder $x$ bidding $1$ in the profile (so $1-\p{x}$ gives the probability that she bids $0$).

\subsection{Generalized Circuits}\label{sec:Gcircuit}

We {recall} a simplified version of  generalized circuits \cite{CDT09} from \cite{FGHLP23,BudishEtAl2023}:

\begin{definition}[Simplified Generalized Circuits]\label{def:Gcircuit} 
A (simplified) generalized circuit $C=(V,\calT)$ is 
  a pair where $V$ is a set of $n$ 
  nodes and $\calT$ is a collection of $m \leq n$
  \emph{gates}.
Each gate $T\in \calT$ is~either~a SUM gate of the form $T=(G_+,u,v,w)$ with three nodes $u,v$ and $w$, or a NOT gate of the form $T=(G_{\neg},u,w)$ with two nodes $u$ and $w$.
In both cases, $w$ is referred to as the output node of gate $T$ and $u,v$ (or just $u$) are input nodes of $T$.
{The output node $w$ is always distinct\footnote{While \cite{BudishEtAl2023} allows the output node $w$ to be used as an input in the same gate, one can create a new node $w'$, use $w'$ as the output instead of $w$, and add two NOT gates to loop the value of $w'$ back to $w$. This way the same hardness result in \Cref{thm:source} holds under this assumption, with a slightly worse approximation parameter $\eps$.} from the input node(s) but in a SUM gate, $u=v$ is allowed. 
Every node $v$ is the output of at most one gate $T_v$. Given $\eps > 0$, an $\eps$-{approximate} solution
to $C$ is an assignment $a \in [0, 1+\eps]^V$
such that for all $v$,
\begin{itemize}
    \item If $T_v=(G_+,u_1,u_2,v)$ then we must have $a_{v}=\min(a_{u_1}+a_{u_2},1) \pm \eps$; and
     \item If $T_v=(G_\neg,u,v)$ then we must have $a_{v}=(1-a_{u}) \pm \eps$.
\end{itemize}}
\end{definition}

We let $\eps$-$\Gcircuit$  denote the problem of computing an $\eps$-approximate solution to a simplified generalized circuit, which is known to be \PPAD-hard even when every node is used as an input of at most $2$ gates. {(In particular, $u$ being used as both inputs in a SUM gate is counted once.)}

\begin{theorem}[Theorem $3$ of \cite{BudishEtAl2023}]\label{thm:source}
There exists a constant $\eps>0$ such that $\eps$-$\Gcircuit$ with gates  $\{G_+,G_{\neg}\}$ is $\PPAD$-hard even when every node is used as an input of at most two gates.
\end{theorem}

\section{Proof of \Cref{thm:HDFPA_ppad_hard}}\label{sec:mainreduction}

\subsection{PPAD-hardness}
Let $\eps$ be a positive constant such that $\eps$-$\Gcircuit$, with SUM and NOT gates (as described above in \Cref{thm:source}), is \PPAD-hard. In this subsection we prove the hardness part of \Cref{thm:HDFPA_ppad_hard} by giving a polynomial-time reduction from $\eps$-$\Gcircuit$ to the problem of finding an $\eta$-NE of a HDFPA, for some sufficiently small positive constant $\eta$ to be specified later.

The HDFPA $H=(\calN,B,v,E)$ we build in the reduction satisfies that
$$
v_i\in \left\{\frac{1}{2}, 2, 3,4\right\},
\quad \text{for every bidder $i$,}
$$
$B=\{0,1\}$,
and every weighted hyperedge $(S,w_e)\in E$ has $|S|\le 4.$ \Cref{thm:HDFPA_ppad_hard} will then follow by dividing all values and bids by $4$ using the normalization argument in \Cref{sec:defHDFPA}. 

Before giving an overview of the reduction, we start by introducing some conventions that will make the construction easier to follow. First, we will reserve symbols $u,v$ and $w$ to denote nodes in the starting generalized circuit $C$, and symbols $x,y,z$ and $h$ to denote bidders in the HDFPA $H$ that we build from $C$.
(So the total weight of hyperedges that contain a bidder $x$ will be denoted by $W_x$.) 
The symbol $n$ will be used to denote the number of nodes in $C$; the number of bidders in the HDFPA $H$ will be linear in $n$, as it becomes clear at the end of the construction.

We will use $p$ to denote a mixed strategy profile of a HDFPA with bidding space $B=\{0,1\}$, where $\p{x}\in [0,1]$ denotes the probability of a bidder $x$ bidding $1$ in the profile (so bidder $x$ bids~$0$ with probability $1-\p{x}$ in the profile $p$). We still use ${\mathbf{b}}\sim p$ to denote the draw of an action profile ${\mathbf{b}}$ (where each bid takes a value in $B=\{0,1\})$ from $p$.

Given a mixed strategy profile $p$ and a bidder $x$, we use $(1,p_{-x})$ to denote the mixed strategy profile obtained from $p$ by making bidder $x$ bid $1$ with probability $1$ (or equivalently, changing $p(x)$ to $1$), and write $(0,p_{-x})$ to denote the profile obtained from $p$ by making $x$ bid $0$ with probability $1$
(or equivalently, changing $p(x)$ to $0$).
For each weighted hyperedge $e$ in the
  weighted hypergraph $E$ of $H$, we write 
$$
U_x^e(b^*,p_{-x}):= \bE_{{\mathbf{b}}_{-x}\sim p_{-x}} \big[w_e \cdot u_x^e(b^*,{\mathbf{b}}_{-x})\big]
$$
to denote the expected utility of bidder $x$ from $e$ when they bid $b^* \in \zo$ and remaining bidders bid according to $p$.
Given a subset $R\subseteq E$, we write $$
U_x^R(b^*,p_{-x})=\sum_{e\in R}U_x^e(b^*,p_{-x})
\quad \text{ and } \quad 
\Delta_{x}^R(p):=U_x^R(1,p_{-x})- U_x^R(0,p_{-x}).%
$$
Note $U_x^R(b^*,p_{-x})$ is the total expected utility bidder $x$ gets from weighted hyperedges in $R$ if she always bids $b^*$ while the remaining bidders bid according to $p_{-x}$. So $\Delta_{x}^R(p)$ captures the advantage of $x$ bidding $1$ against $0$, from weighted hyperedges in $R$ against $p_{-x}$.

We also write $\Delta_{x}(p)$ to denote $\Delta_{x}^E(p)$. 
The following claim follows directly from \Cref{def:normalized-ne}:

\begin{claim}\label{clm:binary-regret}
At any $\eta$-NE $p$ of a HDFPA $H$, the following holds for every bidder $x$:
\begin{align*}
\Delta_{x}(p)>0\implies (1-p(x))\Delta_{x}(p)\leq\eta W_x\quad\text{and}\quad
\Delta_{x}(p)<0\implies -p(x) \Delta_{x}(p)\leq\eta W_x.
\end{align*}
\end{claim}
\begin{proof}
The utility of bidder $x$ when bidders use profile $p$ is $$U_x(p)=p(x)U_x(1,p_{-x})+(1-p(x))U_x(0,p_{-x}).$$ If $\Delta_{x}(p)>0$, changing the mixed strategy of $x$ in $p$ to pure bid $1$ would gain $(1-p(x))\Delta_{x}(p)$.  So the first part of the claim follows from the definition. The other part follows similarly.
\end{proof}

\subsection{Overview of the Reduction}\label{sec:overviewreduction}

Let $C=(V,\calT)$ be a generalized circuit  with $n$ nodes, where every node is used as the input of at most $2$ gates, and $|\calT|=m\le n$.
Given~$C$,
  our goal is to build a HDFPA $H$ in polynomial time such that any $\eta$-NE $p$ of $H$, for some positive constant $\eta$, can be used to recover an $\eps$-approximate solution to $C$ in polynomial time.

We start the construction by adding the following $2n$ bidders to $H$:
For each node $u\in V$, we 
\begin{flushleft}\begin{enumerate}
\item add one bidder $\yy{u}$ with value $4$; we will refer to $\yy{u}$ as the \emph{\color{red}output} bidder for node $u$; and
\item add one {bidder} $\xx{u}$ with value $2$;
  we will refer to $\xx{u}$ as the  \emph{\color{blue}input} bidder for node $u$.
\end{enumerate}\end{flushleft}
We will refer to these $2n$ bidders as the \emph{main} bidders in $H$, while all other bidders that are added later (with values either $1/2$ or $3$) are referred to as \emph{auxiliary} bidders.

Looking ahead, after adding auxiliary bidders and adding  weighted hyperedges to the weighted hypergraph $E$ of $H$, 
our goal is to make sure that
  the following holds for $H$ at the end:
\begin{quote}
\centering\emph{Given any $\eta$-NE $p$ of $H$, 
  assigning $\p{\yy{u}}$ to each\\ node $u\in V$ must be an $\eps$-approximate solution to $C$.}
\end{quote}

This is achieved by adding the following \emph{gadgets} to $H$, each enforcing one constraint over a few main bidders that holds in any $\eta$-NE of $H$
  (we will give more details on the exact meaning of the phrase ``\emph{adding a gadget to $H$}'' after this paragraph):
\begin{flushleft}\begin{enumerate}
\item For each node $u\in V$, add one \emph{COPY} gadget  (see \Cref{sec:copygadget}) from $\yy{u}$ to $\xx{u}$,
  denoted by $R(=,\yy{u},\xx{u})$,
making sure that every $\eta$-NE $p$ of $H$ satisfies $$\p{\xx{u}}\approx \p{\yy{u}};$$
\item For each gate $T=(G_+,u,v,w)$ in $\calT$, add a \emph{SUM} gadget (see \Cref{sec:addgadget}) from $ {\xx{u}}$ and $ \xx{v}$ to $\yy{w}$, denoted by $R(+,  {\xx{u} ,\xx{v} ,\yy{w}})$, making sure that every $\eta$-NE $p$ of $H$ satisfies $$\p{\yy{w}} \approx \min\left(\p{\xx{u}}+\p{\xx{v}},1\right);$$
\item For each gate $T=(G_{\neg},u,w)$ in $\calT$,  add a \emph{NOT} gadget (see \Cref{sec:compgadget}) from $\xx{u}$ to $\yy{w}$, denoted by $R(\neg,\xx{u},\yy{w})$, making sure that every $\eta$-NE $p$ of $H$ satisfies
$$
\p{\yy{w}}\approx 1-\p{\xx{u}},
$$
\end{enumerate}\end{flushleft}
where the approximations above all mean additive errors bounded by $O(\eta)$.
It is easy to verify (see \Cref{sec:piecestogether}) that, by setting $\eta$ to be a sufficiently small constant compared to $\eps$, $(\p{\yy{u}}:u\in V)$~from any $\eta$-NE of $H$ must be an $\eps$-approximate solution to $C$,
  which gives the desired reduction.

The challenge is how to design these three types of gadgets (COPY $=$, SUM $+$, NOT~$\neg$) so that they can work together, enforcing constraints that they are responsible for, while not~interfering with each other.
Towards this end, each gadget (which will be denoted by $R$) consists of not only a number of weighted hyperedges but also a number of auxiliary bidders who \emph{only} participate in hyperedges of $R$. In more {detail}, when we ``add a gadget $R$'' to $H$, we first 
\begin{flushleft}\begin{enumerate}
\item add one \emph{\color{OliveGreen}helper} bidder $\hh{R}$ with value $3$; 
\item add a certain number of bidders $\z{R,1},\z{R,2},\ldots$, where the number depends on the type of the gadget $R$. 
All these bidders have value $1/2$ and are so-called \emph{\color{Sepia}$0$-anchors} (because, due to the no overbidding rule, they have no choice but to bid $0$ with {probability} $1$).
We will refer to bidders $\hh{R},\z{R,1},\z{R,2},\ldots$ as \emph{the auxiliary bidders} of $R$; and then
\item add a number of weighted hyperedges to $E$ in $H$, depending on the type of the gadget, where each hyperedge added 
is a size-at-most-$4$ subset of the following bidders:
\begin{itemize}
\item the auxiliary bidders of $R$: the helper bidder $\hh{R}$ and the $0$-anchors $\z{R,1},\z{R,2},\ldots$; and 
\item \emph{the main bidders} of $R$, i.e., $\yy{u}$ and $\xx{u}$ when $R=R(=,\yy{u},\xx{u})$ is a COPY gadget; \\ $\xx{u},\xx{v}$ and $\yy{w}$ when $R=R(+,\xx{u},\xx{v},\yy{w})$ is a SUM gadget;
or $\xx{u}$ and $\yy{w}$\\ when $R=R(\neg,\xx{u},\yy{w})$ is a NOT gadget.%
\footnote{Note that we're using various colors to make it clear {whether} a bidder is an {\color{blue}input}, {\color{red}output}, {\color{OliveGreen}helper} bidder or a {\color{Sepia}zero-anchor}. Hence, whenever we talk {about} a bidder $x$ or $y$ (with no color), we mean a generic bidder who could be of any type.}
\end{itemize}
We will refer to these weighted hyperedges as \emph{the weighted hyperedges} of $R$;  
we will also abuse the notation to use $R$ to denote the set of weighted hyperedges of $R$ (so notation such as $\Delta_{x}^R$, given a gadget $R$ and a bidder $x$, is well defined).
\end{enumerate}\end{flushleft}

Before describing our first gadget, we record some simple calculations about $\Delta_{x}^e(p)$:

\begin{claim}\label{clm:simple1}
Let $x$ be a bidder in $H$ with value in $\{2,3,4\}$, and $e=(S,w_e)$ be a weighted hyperedge~in $E$ with $S=\{x\}$. Then for any mixed strategy profile $p$ of $H$, we have 
$\Delta_{x}^e(p)=-w_e$.
\end{claim}
\begin{proof}
Let $v_x\in \{2,3,4\}$ be the value of $x$.
Because only bidder $x$ can participate in $e$ (as $S=\{x\}$), we have $U_x^e(1,p_{-x})={v_x\cdot w_e-w_e}$ and $U_x^e(0,p_{-x})={v_x\cdot w_e} $.
The claim follows.
\end{proof}

\begin{claim}\label{clm:delta_contrib}
Let $m$ be a nonnegative integer.
Let $x$ and $y$ be two bidders with
  values $v_x$ and $v_y$ from $\{2,3,4\}$ and $\z{1},\ldots,\z{m}$ be $m$ $0$-anchors in $H$.
Let $e=(S,w_e)$ be a weighted hyperedge in $E$ with $$S=\{x,y,\z{1},\ldots,\z{m}\}$$ (or with $S=\{x,y\}$ when $m=0$).
Then in any mixed strategy profile $p$ of $H$, we have 
    \[
    \Delta^{e}_x(p) = w_e\left(\left(\frac{v_x}{m+2} - \frac{v_x-1}{2}\right)\p{y} + v_x - 1 - \frac{v_x}{m+2}\right).
\]
\end{claim}
\begin{proof}
The $0$-anchors always bid $0$.
If $x$ bids $1$, then either $y$ bids $1$ and they tie, or $x$ wins. Thus, $$U_x^{e}(1,p_{-x}) = w_e\left(\big(1-\p{y}\big)(v_x - 1) + \frac{v_x - 1}{2} \cdot \p{y}\right).$$ If bidder $x$ bids $0$, then either $y$ bids $1$, or every bidder in the auction ties and so $$U_x^{e}(0,p_{-x}) = w_e\big(1-\p{y}\big)\cdot \frac{v_x}{m+2}.$$ %
The claim follows by taking the difference.
\end{proof}
{
Finally, a last claim that will be useful is the following. \begin{claim}\label{claim:bound}
    Let $\x$ be an input bidder (so $\x$ has value $v=2$) and $p$ be a mixed strategy profile. Then for any set of {hyperedges} 
$R$ with $\x \in S$ for all $(S,w_e)=e \in R$, we have 
$$|\Delta_{\x}^R(p)| \leq \sum_{e \in R} w_e{.}$$
\end{claim}
\begin{proof}
Fix $e=(S,w_e)\in R$ and let $\bk$ be the number of bidders in $S\setminus\{\x\}$ that bid $1$ under $p_{-\x}$. If $\x$ bids $1${,} she is always a highest bidder, tied with $\bk$ others, so $U^e_{\x}(1,p_{-\x})=w_e \cdot q_1$ with $q_1:=\E[1/(1+\bk)]$. If $\x$ bids $0${,} she is a highest bidder only when $\bk=0$, in which case all $|S|$ bidders tie, so $U^e_{\x}(0,p_{-\x})=2w_e \cdot q_0$ with $q_0:=\Pr[\bk=0]/|S|\le \Pr[\bk=0]\le q_1$. Hence $$-w_e\le -w_eq_1\le w_e\cdot (q_1-2q_0)=\Delta^e_{\x}(p)\le w_eq_1\le w_e{.}$$ 
Thus, summing over all edges in $R$ proves the {claim}.

\end{proof}
}

\subsection{COPY Gadget}\label{sec:copygadget}

Let $\Lambda>0$ be a parameter, which we will fix to be $40$ later. 

Let $\y$ be one of the output bidders (with value $4$) and $\x$ be one of the input bidders (with value $2$).
Adding a COPY gadget $R=R(=,\y,\x)$ to the HDFPA $H$ means the following:
\begin{flushleft}\begin{enumerate}
\item add one helper bidder $\h$ with value $3$;
\item add three $0$-anchors $\z{1},\z{2}$ and $\z{3}$ with value $1/2$ (note that we skip the subscript $R$ in names of these auxiliary bidders of $R$ but it should be understood that they are auxiliary bidders of $R$ and only participate in weighted hyperedges of $R$ to be added below);
\item add the following seven weighted hyperedges to $E$:
\begin{center}
\begin{tabular}{c|@{\qquad}c@{\qquad}c@{\qquad}c@{\qquad}c@{\qquad}c@{\qquad}c@{\qquad}c}
\toprule
\text{hyperedge}
& $\{\y,\h\}$
& $\{\y,\h,\z{1}\}$
& $\{\y\}$
& $\{\x,\h\}$
& $\{\x,\h,\z{2},\z{3}\}$
& $\{\h\}$
& $\{\x\}$ \\[0.8ex]
\text{weight}
& $4\Lambda$
& $12\Lambda$
& $24\Lambda$
& $4\Lambda$
& $16\Lambda$
& $36\Lambda$
& $9\Lambda$ \\
\bottomrule
\end{tabular}
\end{center}
\end{enumerate}\end{flushleft}
\vspace{.4cm}
We abuse the notation to use $R$ to denote the set of these seven weighted hyperedges.

Now we prove a lemma which shows that under a very mild condition about the final HDFPA $H$, any of its $\eta$-NE satisfies $\p{\x}\approx \p{\y}$.
To this end, we let $H$ be any HDFPA that (1) contains bidders $\x,\y,\h,\z{1},\z{2},\z{3}$ with values given earlier; (2) contains a copy of the seven weighted hyperedges in $R$ as given in the table above; (3) does not have any other weighted hyperedge that contains $\h,\z{1},\z{2}$ or $\z{3}$; and (4)  the total weight of all other hyperedges in $H$ (excluding those in $R$) that contain $\x$ is at most ${\Lambda/2}$. Then the following lemma holds:%

\begin{lemma}\label{lem:copy-identities}
Let $H$ be a HDFPA that satisfies conditions $(1)$-$(4)$ above. Then, any $\eta$-NE $p$ of $H$ satisfies 
\begin{equation*}\label{eq:copy-bound}
\big|\p{\x}-\p{\y}\big|\leq
300\eta.%
\end{equation*}
Moreover, any mixed strategy profile $p$ of $H$ satisfies $\Delta_{\y}^R(p)=0$.
\end{lemma}
\begin{proof}
Given any mixed strategy profile $p$ of an $H$ that satisfies  conditions (1)-(4) above, we have
\begin{equation*}
\Delta_{\y}^R(p)=0,
\quad
\Delta_{\x}^{R}(p)=2\Lambda\left(\p{\h}-\frac12\right)
\quad\text{and}\quad
\Delta_{\h}^{R}(p)=2\Lambda\big(\p{\y}-\p{\x}\big).
\end{equation*}

The above equalities follow using \Cref{clm:simple1} and \Cref{clm:delta_contrib}.
For $\y$, we have
\[
    \Delta_{\y}^{R}(p) = 4\Lambda\left(\frac{\p{\h}}{2} +1 \right) + 12\Lambda\left(\frac{5}{3} - \frac{\p{\h}}{6} \right) - 24\Lambda = 0.
    \]
    Similarly, for $\x$ we can compute 
    \[\Delta_{\x}^{R}(p) = 4\Lambda\left(\frac{
    \p{\h}}{2} \right) + 16\Lambda\left(\frac12 \right) - 9\Lambda = 2\Lambda\left(\p{\h}-\frac12\right).\]
    Finally, for $\h$ we can compute 
    \[\Delta_{\h}^{R}(p) = 4\Lambda\left( \frac{\p{\y}}{2} + \frac12\right) + 12\Lambda + 4\Lambda\left(\frac{\p{\x}}{2} + \frac12\right) + 16\Lambda\left(\frac54 - \frac{\p{\x}}{4} \right) -36\Lambda = 2\Lambda\big(\p{\y} - \p{\x}\big).\]

Suppose that $p$ is an $\eta$-NE of $H$ and assume for a contradiction that $\delta:=|\p{\y}-\p{\x}|>300\eta$.

Consider the first case when $\delta= \p{\x} - \p{\y} > 0$. Because no other hyperedge contains $\h$, we have (recall that we use $\Delta_{\h}$, without a superscript, to denote the overall $\Delta$ of $\h$ from the whole $H$) $$\Delta_{\h}(p) = \Delta_{\h}^{R}(p) = -2\Lambda \delta < -600 \Lambda \eta \leq  0.$$ As the total weight of $\h$ is $72 \Lambda$, we have $\p{\h}<1/8$ from \Cref{clm:binary-regret}.

Now we examine $\x$. Given that the total weight of hyperedges outside of $R$ that contain $\x$ is at most $\Lambda/2$, by \Cref{claim:bound}, the $\Delta$ of $\x$ from these hyperedges has absolute value at most $\Lambda/2$. %
Thus,
$$
\Delta_{\x}(p)\le 2\Lambda \left(\frac{1}{8}-\frac{1}{2}\right)+\frac{\Lambda}{2}\le -\frac{\Lambda}{4}.
$$
However, given that the total weight of $\x$ is at most $29\Lambda+\Lambda/2<30\Lambda$, we have from \Cref{clm:binary-regret} that $\p{\x}\le 120\eta$, contradicting the assumption that $\p{\x}=\p{\y}+\delta\ge \delta>300\eta$.

Consider the other case that $\delta= \p{\y} - \p{\x} > 0$. Then, $\Delta_{\h}(p)=2\Lambda \delta > 600\Lambda\eta \geq 0$ and so by \Cref{clm:binary-regret}, we have $(1 - \p{\h})\le 1/8$ and thus, $\p{\h}\ge 7/8$.
Using similar arguments, we have 
$$\Delta_{\x}(p)\ge \frac{3\Lambda}{4}-\frac{\Lambda}{2}=\frac{\Lambda}{4}.$$
This implies by \Cref{clm:binary-regret} that $1-\p{\x}\le 120\eta$, {contradicting} $\p{\x}=\p{\y}-\delta \leq 1-\delta<  1-300\eta$.
\end{proof}

\subsection{SUM Gadget}\label{sec:addgadget}

Let $\xx{1}$ and $\xx{2}$ be two input bidders (with value $2$) and $\y$ be an output bidder (with value $4$). \footnote{Note that we allow $\xx{1}=\xx{2}$ since the same vertex can be used as both inputs of a SUM gate.} 
Adding a SUM gadget $R=R(+,\xx{1},\xx{2},\y)$ to the HDFPA $H$ means the following:
\begin{flushleft}\begin{enumerate}
\item add one helper bidder $\h$ with value $3$;
\item add six $0$-anchors $\z{1},\ldots,\z{6}$ (so $\h,\z{1},\ldots,\z{6}$ are auxiliary bidders of $R$);
\item add the following six weighted hyperedges to $E$:
\begin{center}
\begin{tabular}{c|@{\qquad}c@{\qquad}c@{\qquad}c@{\qquad}c@{\qquad}c@{\qquad}c}
\toprule
\text{hyperedge}
& $\{\xx{1},\h,\z{1},\z{2}\}$
& $\{\xx{2},\h,\z{3},\z{4}\}$
& $\{\y,\h\}$
& $\{\y,\h,\z{5},\z{6}\}$
& $\{\h\}$
& $\{\y\}$ \\[0.8ex]
\text{weight}
& $4$
& $4$
& $8$
& $12$
& $29$
& $31$ \\
\bottomrule
\end{tabular}
\end{center}
\vspace{.4cm}

\end{enumerate}\end{flushleft}
We abuse the notation to use $R$ to denote the set of these six weighted hyperedges.

We prove a lemma which shows that, under a very mild condition about the final HDFPA $H$, any of its $\eta$-NE $p$ satisfies $\p{\y}\approx \min(\p{\xx{1}}+\p{\xx{2}},1)$.
To this end, we let $H$ be any HDFPA~that (1) contains bidders $\y,\xx{1},\xx{2},\h,\z{1},\ldots,\z{6}$ with values given earlier; (2) contains a copy of the~six weighted hyperedges in $R$ given above; (3) does not have any other hyperedges that {contain} $\h,\z{1},$ $\ldots,$ or $\z{6}$; and (4) in any mixed strategy profile $p$ of $H$, we have $\Delta_{\y}(p)=\Delta_{\y}^R(p)$ (looking back,  this comes from the zero contribution of the COPY gadget to $\Delta_{\y}(p)$,
as shown in \Cref{lem:copy-identities}).

Note that unlike the situation with the COPY gadget, we do not assume that the total weight of hyperedges that contain $\y$ (excluding those from $R$) is small compared to weights of hyperedges in $R$.
Indeed we will write $\Gamma$ to denote an upper bound on the total weight of hyperedges (including those in $R$) that contain $\y$, and 
$\Gamma$ in our application later will be a relatively large constant (compared to the total weight of $\y$ from $R$, which is $51$).
With these assumptions, we prove the following:

\begin{lemma}\label{lem:plus-identities}
Any mixed strategy profile $p$ of a HDFPA $H$ that satisfies the conditions above has 
\begin{equation}\label{eq:plus-identities}
\Delta_{\h}^{R}(p)= \p{\y}-\p{\xx{1}}-\p{\xx{2}}\quad\text{and}\quad
\Delta_{\y}^{R}(p)= 1-2\p{\h}.
\end{equation}
Moreover, if $p$ is an $\eta$-NE of $H$, then we must have
$$
\left|\p{\y}-\min\big(\p{\xx{1}}+\p{\xx{2}},1\big)\right|\le \max\big(300,2\Gamma\big)\cdot \eta.
$$
\end{lemma}
\begin{proof}
The first part of the lemma is just calculation using \Cref{clm:simple1} and \Cref{clm:delta_contrib}.

For $\h$ of value $3$, we have
\[
\Delta_{\h}^R(p)=4\left(\frac54-\frac{\p{\xx{1}}}{4}\right)
+4 \left(\frac54-\frac{\p{\xx{2}}}{4}\right)+
8\left(\frac12+\frac{\p{\y}}{2}\right)
+12\left(\frac54-\frac{\p{\y}}{4}\right)-29
\]
which is equal to $\p{\y} - \p{\xx{1}}-\p{\xx{2}}$.

For $\y$ of value $4$, we have
\[
\Delta_{\y}^R(p)=8 \left(1+\frac{\p{\h}}{2}\right)
+12 \left(2-\frac{\p{\h}}{2}\right)-31
=1 -2 \p{\h}.
\]

Suppose $p$ is an $\eta$-NE. Assume for a contradiction that $$\delta:=\left|\p{\y}-\min\big(\p{\xx{1}}+\p{\xx{2}},1\big)\right|>\max\big(300,2\Gamma\big)\cdot \eta.$$

Consider the first case that 
  $\p{\y}=\min(\p{\xx{1}}+\p{\xx{2}},1)+\delta$.
Given that $\p{\y}\in [0,1]${,} it~must~be that $\p{\y}=\p{\xx{1}}+\p{\xx{2}}+\delta$ and thus,
  $\Delta_{\h}(p)=\Delta_{\h}^R(p)\ge 300\eta$.
Given that the total weight of $\h$ is $57$, using \Cref{clm:binary-regret}, we have $1-\p{\h}\le 1/5$ and $\p{\h}\ge 4/5$.
This implies  $\Delta_{\y}(p)=\Delta_{\y}^R(p)\le -3/5$.
Given that the total weight of $\y$ is at most $\Gamma$, we have from \Cref{clm:binary-regret} that $\p{\y}\le 5\Gamma\eta/3 \leq 2\Gamma\eta< \delta$, a contradiction with $\p{\y}\ge \delta$.

Consider the other case of 
  $\p{\y}=\min(\p{\xx{1}}+\p{\xx{2}},1)-\delta$.
Using $\p{\y}\le \p{\xx{1}}+\p{\xx{2}}-\delta$,~we have $\Delta_{\h}(p)\le -300\eta$ and thus, $\p{\h}\le 1/5$.
This implies $\Delta_{\y}(p)\ge 3/5$, and  $1-\p{\y}\leq 2\Gamma \cdot \eta< \delta$, %
which {contradicts} the assumption that $\p{\y}\le 1-\delta$.
\end{proof}

\subsection{NOT Gadget}\label{sec:compgadget}

Finally let $\x$ be an input bidder  (with value $2$) and let $\y$ be an output bidder (with value $4$). 
Adding a NOT gadget $R=R(\neg,\x ,\y)$ to the HDFPA $H$ means the following:
\begin{flushleft}\begin{enumerate}
\item add one helper bidder $\h$ with value $3$;
\item add two $0$-anchors $\z{1}$ and $\z{2}$ (so $\h,\z{1}$ and $\z{2}$ are auxiliary bidders of $R$);
\item add the following five weighted hyperedges to $E$:
\begin{center}
\begin{tabular}{c|@{\qquad}c@{\qquad}c@{\qquad}c@{\qquad}c@{\qquad}c}
\toprule
\text{hyperedge}
& $\{\x,\h\}$
& $\{\y,\h\}$
& $\{\y,\h,\z{1},\z{2}\}$
& $\{\h\}$
& $\{\y\}$ \\[0.8ex]
\text{weight}
& $2$
& $8$
& $12$
& $21$
& $31$ \\
\bottomrule
\end{tabular}
\end{center}
\end{enumerate}\end{flushleft}
\vspace{.4cm}
We abuse the notation and use $R$ to denote the set of these five weighted hyperedges.

We prove a lemma which shows that, under a very mild condition about the final HDFPA $H$, any of its $\eta$-NE $p$ satisfies $\p{\y}\approx 1-\p{\x}$.
To this end, we let $H$ be any HDFPA~that (1) contains bidders $\x,\y ,\h,\z{1},\z{2}$ with values given earlier; (2) contains a copy of the~five weighted hyperedges in $R$ given above; (3) does not have any other hyperedges that {contain} $\h,\z{1}$ or $\z{2}$; and (4) {is such that,} in any mixed strategy profile $p$ of $H$, we have $\Delta_{\y}(p)=\Delta_{\y}^R(p)$.
Similarly, we write $\Gamma$ to denote an upper bound on the total weight of hyperedges (including those in $R$) that contain $\y$.

We prove the following lemma:

\begin{lemma}\label{lem:complement-identities}
Any mixed strategy profile $p$ of an $H$ that satisfies the conditions above has 
\begin{equation}\label{eq:complement-identities}
\Delta_{\h}^R(p)= \p{\y}-\big(1-\p{\x}\big)
\quad\text{and}\quad
\Delta_{\y}^R(p)= 1-2\p{\h}.
\end{equation}
Moreover, if $p$ is an $\eta$-NE of $H$, then we must have
$$
\left|\p{\y}-\big(1-\p{\x}\big)\right|\le \max(300,2\Gamma)\cdot \eta.
$$
\end{lemma}

\begin{proof}
For $\h$ of value $3$, we have 
\[
\Delta_{\h}^R(p)=
2 \left(\frac12+\frac{\p{\x}}{2}\right)+8\left(\frac12+\frac{\p{\y}}{2}\right) + 12\left(\frac{5}{4}-\frac{\p{\y}}{4}\right)-21
=\p{\y}+\p{\x}-1.
\]
For $\y$ of value $4$, we have 
$$
\Delta_{\y}^R(p)=8\left(1+\frac{\p{\h}}{2}\right)+12\left(2-\frac{\p{\h}}{2}\right)-31=1-2p(\h).
$$
The rest of the proof is similar to 
  that of \Cref{lem:plus-identities}.
\end{proof}

\subsection{Putting Pieces Together}\label{sec:piecestogether}

We are now ready to prove the \PPAD-hardness part of \Cref{thm:HDFPA_ppad_hard}:

\begin{proof}[Proof of \PPAD-hardness for \Cref{thm:HDFPA_ppad_hard}]
Given {a generalized} circuit $C=(V,\calT)$ with $n$ nodes and where every node is the input of at most $2$ gates, we construct a HDFPA $H$ as described earlier in \Cref{sec:overviewreduction}, with $\Lambda=40$. 
Set $\Gamma=2000$.
The following constraints hold in every $\eta$-NE $p$ of $H$:
\begin{enumerate}
\item Trivially we have $\p{\yy{u}}\in [0,1]$ for every node $u\in V$;
\item For every node $u\in V$, we have 
  $\p{\xx{u}}=\p{\yy{u}} \pm 300\eta$;
\item For every SUM gate $(G_+,u,v,w)\in \calT$, we have 
$\p{\yy{w}}=\min\big(\p{\xx{u}}+\p{\xx{v}},1\big)\pm 2\Gamma\eta$;
\item For every NOT gate $(G_\neg,u,w)\in \calT$, we have 
$\p{\yy{w}}=1-\p{\xx{u}}\pm 2\Gamma\eta.
$
\end{enumerate}
It is easy to show that $(\p{\yy{u}}:u\in V)$ satisfies
\begin{enumerate}
\item For every SUM gate $(G_+,u,v,w)\in \calT$, we have 
\begin{align*}
\p{\yy{w}}&=\min\big(\p{\xx{u}}+\p{\xx{v}},1\big)\pm 2\Gamma\cdot \eta\\ &=\min\big(\p{\yy{u}}+p(\yy{v})\pm 600\eta ,1\big)\pm 2\Gamma\cdot \eta\\
&= \min\big(\p{\yy{u}}+p(\yy{v}) ,1\big)\pm 3\Gamma\cdot \eta;
\end{align*}
\item For every NOT gate $(G_\neg,u,w)\in \calT$, we have
\begin{align*}
\p{\yy{w}}=1-\p{\xx{u}}\pm 2\Gamma\eta
=1-\big(\p{\yy{u}}\pm 300\eta\big)\pm 2\Gamma\cdot \eta=1-\p{\yy{u}}\pm 3\Gamma\cdot \eta.
\end{align*}
\end{enumerate}
So $(\p{\yy{u}}:u\in V)$ is {a} $(3\Gamma \cdot \eta)$-{approximate} solution to $C$ which can be made {an} $\eps$-{approximate} solution to $C$ by choosing $\eta$ to be a small constant. 
Finally, the theorem  follows {by normalizing bids and values as pointed out in \Cref{sec:defHDFPA}}. In this case, we would divide all bids and values by $4$ and get \PPAD-hardness of finding an $\eta/4$-NE of a HDFPA for some constant $\eta$ where every hyperedge has at most $4$ participants, $B=\{0,1/4\}$ and every bidder has a value in $\{1/8, 1/2, 3/4,1\}.$
\end{proof}

\subsection{Membership in $\PPAD$}
To prove membership in $\PPAD$ we will need the following lemma:
\begin{lemma}\label{lem:simple1}
Given a HDFPA $H$ and a mixed strategy profile $\sigma$,
the expected utility $U_i(\sigma)$ of any bidder $i$ can be computed
in time polynomial in the length of
$H$ and $\sigma$.
\end{lemma}

\begin{proof}
Fix a weighted hyperedge $e=(S,w_e)$ such that $i \in S$, and say bidder $i$ submits a bid
$b\in B$ and has value $v_i$. We denote the bidders in $S \setminus \{i\}$ as  $\{j_1,\ldots,j_m\}$. For each $r \leq m$ we define:
$$
\ell_r:=\Pr_{\mathbf{a}\sim\sigma_{j_r}}[\mathbf{a}<b] \quad \text{ and } \quad 
t_r:=\Pr_{\mathbf{a}\sim\sigma_{j_r}}[\mathbf{a}=b].
$$
Clearly $\ell_r$ and $t_r$ can be computed in polynomial time in the length of $\sigma$'s representation.

For $0\leq k\leq r\leq m$, let $D_{r,k}$ be the probability that
exactly $k$ of the bidders $j_1,\ldots,j_r$ bid $b$ and all remaining
bidders among them bid strictly below $b$. We can compute these probabilities using dynamic programming in polynomial time. We define 
$$D_{0,0}=1 \quad \text{ and } \quad D_{r,k}=0 \text{ whenever  $k<0$ or $k>r$}.$$ Since bidders choose their bids independently, we have
$$D_{r,k} = \ell_r \cdot  D_{r-1,k} + t_r \cdot D_{r-1,k-1}.$$
If any of $j_1, \ldots, j_m$ bids above $b$, bidder $i$ loses. Otherwise, if $k$ bidders among $j_1, \ldots, j_m$ bid $b$ and all other bidders bid below
$b$, then bidder $i$ wins with probability $1/(k+1)$. Hence,
\[
\mathbb{E}_{\mathbf b_{-i}\sim\sigma_{-i}}
\left[u_i^e(b,\mathbf b_{-i})\right]
=
(v_i-b)\sum_{k=0}^{m}\frac{D_{m,k}}{k+1}.
\]
The dynamic program has $O(m^2)$ entries, and each entry can be computed in polynomial time.~So we can compute the above expression in polynomial time. 

We have that bidder $i${'s} expected utility is:
\[
U_i(\sigma)
=
\sum_{\substack{e=(S,w_e) \in E\\ i\in S}} w_e
\sum_{b\in B }
\sigma_i(b)\,
\mathbb{E}_{\mathbf b_{-i}\sim\sigma_{-i}}
\left[u_i^e(b,\mathbf b_{-i})\right],
\]
which shows that $U_i(\sigma)$ can be computed in polynomial time.
\end{proof}

We now turn to the proof of $\PPAD$ membership. In particular, we prove the following:
\begin{lemma}\label{lem:HDFPA_in_PPAD}
    The problem of finding an $\eta$-NE of an HDFPA $H$, when $\eta>0$ is a rational number given in binary in the input,
is in \PPAD.
\end{lemma}
\begin{proof}

 Let $H=(\calN,B,v,E)$. First we normalize utilities of bidder $i$ by defining (recall $W_i>0$)
\[
\overline{ u}_{i}(b)=\frac{u_i(b)}{W_i}.
\]
Given a mixed strategy profile $\sigma$, we write 
$$
\overline{U}_i(\sigma)
=\frac{U_i(\sigma)}{W_i}
$$
to denote the normalized expected payoff of bidder $i$.
Then we have (1) the $n$-bidder normal-form game,
denoted by $G$ below, defined using $\overline{u}_i(\cdot)$, $i\in [n]$, has $0\le \overline{u}_i(b)\le 1$
  for every $b$ and $i$; and 
(2) any (additive) $\eta$-approximate Nash equilibrium of $G$, i.e., a mixed strategy profile $\sigma$ satisfying
$$
\overline{U}_i(\sigma)\ge \overline{U}_i(b,\sigma_{-i})-\eta,\quad\text{for any $i$ and $b\in B_i$,}
$$
is equivalently an $\eta$-NE of $H$. So it suffices to show that the problem of finding an $\eta$-approximate Nash equilibrium of $G$, where
  $\eta$ is part of the input, is in \PPAD.

For this purpose, we combine Nash's Brouwer map in \cite{Nash} and a general result from \cite{EtessamiYannakakis} (Proposition 2.2).
Recall that given any normal-form game, such as the $n$-bidder game $G$ we are interested in here, Nash used the following map $f$ to prove the existence of a Nash equilibrium: $f$ maps any strategy profile $\sigma$ to a new strategy profile $\sigma'=f(\sigma)$, where
\begin{equation}\label{eq:nash-map}
\sigma_{i,b}'
=\frac{\sigma_{i,b}+\phi_{i,b}(\sigma)}{1+\sum_{a\in B_i} \phi_{i,a}(\sigma)},\quad\text{where
$\phi_{i,a}(\sigma)=\max\big\{0,\overline{U}_i (a,\sigma_{-i})-\overline{U}_i (\sigma)\big\}$.}
\end{equation}
Nash showed that every fixed point $\sigma$ of $f$ is an exact Nash equilibrium; Etessami and Yannakakis showed that any
  $\eta'$-weakly approximate fixed point $\sigma$ of $f$ (i.e., $\|f(\sigma)-\sigma\|_\infty\le \eta'$) with
$$
\eta'=\frac{\eta^2}{16 |B|^3}
$$
is an $\eta$-approximate NE of $G$ (see Proposition 2.3, by setting $M$ there to be $1$ given that payoffs of $G$ are normalized to $[0,1]$).

We would like to apply Proposition 2.2 of \cite{EtessamiYannakakis} to conclude that the problem of finding~an $\eta'$-weakly approximate fixed point of $f$ is in \PPAD. It suffices to show that $f$ is polynomially~computable (see \Cref{lem:simple1}) and polynomially continuous.
For the latter, 
we note the following arguments that work for general games with payoffs normalized in $[0,1]$:
Given any two mixed strategy profiles $\sigma$ and $\tau$ with $\|\sigma-\tau\|_\infty\le\xi$,  it is easy to show using total variation distance that 
\[
\big|\overline{U}_i(\sigma)-\overline{U}_i(\tau)\big|\le \frac{n|B|\xi}{2}\quad\text{and}\quad
\big|\overline{ U}_i(a,\sigma_{-i})-\overline{U}_i(a,\tau_{-i})\big|
\le \frac{n|B|\xi}{2},
\]
from which we have $|\phi_{i,a}(\sigma)-\phi_{i,a}(\tau)|\le n|B|\xi.$
It then follows from standard arguments that 
$$\big|\sigma_{i,b}'-\tau_{i,b}'\big| \leq 4 n|B|^2\xi,$$
from which we have that $f$ is polynomially continuous in the sense of \cite{EtessamiYannakakis}.
\end{proof}

\section{Proof of \Cref{thm:DFPA_ppad_hard}}\label{sec:maintheorem}

\subsection{$\PPAD$-hardness}\label{sec:finalconst}

{ We first recall the proof of \Cref{thm:HDFPA_ppad_hard}. We showed that there is a constant $\eta$ such that computing an $\eta$-NE of an HDFPA is \PPAD-hard. The hard instances satisfy the following conditions:
\begin{flushleft}\begin{itemize}
    \item[(a)] Every {hyperedge} has at most $4$ bidders; and  
    \item [(b)] Each bidder $i$ has value  $v_i \in \{1/2, 2,3,4\}$ and the bidding set is $B=\{0,1\}$.     
\end{itemize}\end{flushleft}
}

Consider such {a} hard HDFPA instance $H=(\calN, B,v,E)$. We build the following DFPA $A_H=(\calN', B',\calV, \calD)$. The new bidding set is $B' = \{0,1,2\}$, and
we add the following bidders to $A_H$:
\begin{flushleft}\begin{itemize}
    \item For each bidder $i$ in $H$ {we} create a bidder $x_i$ in $A_H$. 
    (As it becomes clear in the construction of the value distribution $\calD$ below, 
    $x_i$'s value drawn from $\calD$ will be either $0$ or $v'_i \coloneqq v_i + 1$.)
    \item Add 50 ``dummy bidders''. (They will always have  value $0$ from $\calD$.)
\end{itemize}\end{flushleft}
We specify the value distribution $\calD$ by the following points with \emph{unnormalized} mass:
\begin{flushleft}\begin{itemize}
    \item For each hyperedge $e=(S,w_e) \in E$, add the following point to the distribution $\calD$ with mass $w_e$: for each $i \in S$, bidder $x_i$ receives value $v'_i$; all other bidders receive value $0$.\footnote{We assume that there are no hyperedges $e=(S,w_e)$ and $e'=(S',w_{e'})$ with $S=S'$, as these can be merged into a new hyperedge $e^{*}=(S,w_e+w_{e'})$.}
    \item For each bidder $i$ in $H$ with original value $v_i = 1/2$ (recall that these were called {$0$-anchors}), add a point with mass $W_i$:  $x_i$ receives value $3/2$ and all other bidders receive value $0$. 
\end{itemize}\end{flushleft}
We normalize the above masses to obtain the final probability distribution $\calD$. Importantly, in any $v \in \supp(\calD)$ there are at most $4$ bidders who have non-zero value. Given the description of {a} hard HDFPA instance $H$, we can in polynomial time build the DFPA $A_H$.

We will now show that given any $\eps$-approximate MBNE of $A_H$, we can obtain an $O(\eps)$-NE of $H$ in polynomial time.  
First, the following lemma shows that at any $\eps$-approximate MBNE  of $A_H$, a bidder $x_i$ only bids $0$ with small probability when they have the non-zero value $v_i'$. 
\begin{lemma}\label{lem:small_0}
    Let $\eps \geq 0$ and $\beta$ be an $\eps$-approximate MBNE of $A_H$. For any {bidder} $i$ of $H$, we have
    \[\beta_{x_i}(v^\prime_i)(0) \leq 10\eps.\]
\end{lemma}

We delay its proof to the end of the section and {prove} the \PPAD-hardness part of \Cref{thm:DFPA_ppad_hard}.

    Let $\eps'=\eta/121$ and $\beta$ be any $\eps'$-approximate MBNE of $A_H$. We will show that we can obtain from $\beta$ in polynomial {time}   an $\eta$-NE of $H$. Consider the profile $\sigma$ for $H$, where each bidder $i$ gets  $$\sigma_i(1) = \beta_{x_i}(v^\prime_i)(2) \quad \text{ and } \quad \sigma_i(0)=1-\sigma_i(1).$$
We will also consider the profile $\hat{\beta}$ for $A_H$ which corresponds to $\sigma$. In particular we set 
    $$\hat{\beta}_{x_i}(v_i')(2)={\beta}_{x_i}(v_i')(2)
\quad\text{ and }\quad\hat{\beta}_{x_i}(v_i')(1)=\beta_{x_i}(v_i')(0)+\beta_{x_i}(v_i')(1) \quad\text{ and }\quad \hat{\beta}_{x_i}(v_i')(0)=0{.}$$
Note that any bidder $i$  of $H$ with $v_i=1/2$ will have $\sigma_i(0)=1$,
satisfying the no-overbidding rule.
This is because $\beta$ must satisfy the no-overbidding rule and thus, $\beta_{x_i}(v_i')(2)=0$. 
For the dummy bidders, both $\beta$ and $\smash{\hat{\beta}}$ have them always bid $0$ given the no-overbidding rule.

We verify that $\sigma$ is indeed an $\eta$-NE of $H$ in the rest of the proof. 
This is trivial for bidders $i$ of $H$ with $v_i=1/2$ since bidding $0$ is their only choice given the no-overbidding rule. 

Next consider a bidder $i$ of $H$ with value $v_i\ge 2$. 
Fix {any} bid $b^* \in \{0,1\}$. We consider an edge $e=(S,w_e)$ of $H$ with $i \in S$. Consider a bid profile $b$ in $A_H$ where $x_i$ submits $b^*+1$, every {bidder} $x_j$ with $j \in S$ {submits} a bid $b_j \in \{1,2\}$ and remaining bidders submit $0$. Let $b'$ be a profile for $H$ where bidder $\ell \in S$ bids $b_\ell-1$ and bidders outside $S$ bid $0$. Then, we have $$u_{x_i}(b;v'_i)=u_{i}^e(b')$$
where $u_{x_i}$ refers to the utility of bidder $x_i$ in $A_H$, while $u^e_{i}$ refers to the utility of bidder $i$ coming from the hyperedge $e$.  As such we have:
\begin{align*}
    U_i(b^*,\sigma_{-i}) &= \sum_{\substack{(S,w_e) \in E\\ i \in S} } w_e \E_{{\mathbf{b}}_{-i} \sim \sigma_{-i} } \left[ u^{e}_i\left(b^*, {\mathbf{b}}_{-i}\right)\right]\\ &= W_i \E_{{{\mathbf{v}}_{-x_i}} \sim \calD_{x_i|v^\prime_i}} \E_{\mathbf{b}_{-x_i} \sim {\hat{\beta}_{-{x_i}}}({{\mathbf{v}}_{-x_i}})}\left[u_{x_i}\left(b^* + 1,\mathbf{b}_{-x_i};v_i'\right)\right]{.} 
\end{align*}
Note that the expected utility for $x_i$ can only increase under $\beta_{-x_i}$ compared to $\hat{\beta}_{-x_i}$, as opponents may bid $0$ instead of $1$: 
\begin{align*}
&\E_{{{\mathbf{v}}_{-x_i}} \sim \calD_{x_i|v^\prime_i}} \E_{\mathbf{b}_{-x_i} \sim {\hat{\beta}_{-{x_i}}}({{\mathbf{v}}_{-x_i}})}\left[u_{x_i}\left(b^* + 1,{\mathbf{b}}_{-x_i};v_i'\right)\right] \\[0.3ex]
&\hspace{1cm}\leq \E_{{{\mathbf{v}}_{-x_i}} \sim \calD_{x_i|v^\prime_i}} \E_{\mathbf{b}_{-x_i} \sim {{\beta}_{-{x_i}}}({{\mathbf{v}}_{-x_i}})}\left[u_{x_i}\left(b^* + 1,{\mathbf{b}}_{-x_i};v_i'\right)\right].
\end{align*}
We now compare the expected utility of $x_i$ under $\hat\beta$ with that under $\beta$, in two steps.

First we consider $(\hat\beta_{x_i},\beta_{-x_i})$, in which $x_i$ bids $1$ whenever she would have bid $0$ under $\beta$ and all other bidders play $\beta_{-x_i}$. We claim that this can only increase the utility of $x_i$. Fix any point $v_{-x_i}\in\supp(\calD_{x_i|v_i'})$ and any possible profile $b_{-x_i}$ under $\beta_{-x_i}$. Note that $(v_{x_i},v_{-x_i})$ corresponds naturally to some hyperedge $e=(S,w_e)$ of $H$ with $i \in S$. If some other bidder bids $2$, then $x_i$ gets utility $0$ whether she bids $0$ or $1$. Otherwise, bidding $1$ makes $x_i$ a highest bidder, tied with at most $3$ other bidders (only the bidders $x_j$ with $j\in S\setminus\{i\}$ have nonzero value and can bid $1$; every bidder with value $0$ bids $0$), so she receives at least $(v_i'-1)/4\ge 1/2$; whereas bidding $0$ makes her tie with the $50$ dummy bidders (and possibly others), so she receives at most $v_i'/50\le 1/10$. Hence $u_{x_i}(1,b_{-x_i};v_i')\ge u_{x_i}(0,b_{-x_i};v_i')$ for every possible $b_{-x_i} \in \supp(\beta_{-x_i}(v_{-x_i}))$. Therefore the expected utility of $x_i$ under the profile $(\hat\beta_{x_i},\beta_{-x_i})$ is at least her expected utility under $\beta$.

We now consider the full profile $\hat{\beta}$. To understand how the utility of $x_i$ changes, consider the following: First draw ${\mathbf{b}}_{-x_i}\sim\beta_{-x_i}(v_{-x_i})$ and replace every bid $0$ of a bidder with nonzero value by $1$; the resulting profile is distributed according to $\hat\beta_{-x_i}(v_{-x_i})$. The two profiles differ only if some $x_j$ with $j\in S\setminus\{i\}$ bids $0$ under $\beta$. Since $v_{-x_i}$ has at most $3$ bidders with nonzero value other than $x_i$, and each of them bids $0$ with probability at most $10\eps'$ by \Cref{lem:small_0}, this happens with probability at most $30\eps'$ by a union bound. Moreover, since $x_i$ bids at least $1$ under $\hat\beta_{x_i}$, her utility always lies in $[0,v_i'-1]\subseteq[0,4]$, so on the event that the two profiles differ her utility changes by at most $4$. Hence the expected utility of $x_i$ under $\hat\beta$ is at least her expected utility under $(\hat\beta_{x_i},\beta_{-x_i})$ minus $4\cdot 30\eps'=120\eps'$. 
All together, we obtain:%
\[
   \E_{{{\mathbf{v}}_{-x_i}} \sim \calD_{x_i|v^\prime_i}} \E_{\mathbf{b} \sim {\hat{\beta}}(v_i',{{\mathbf{v}}_{-x_i}})}\left[u_{x_i}\left({\mathbf{b}};v_{i}'\right)\right]  \geq \E_{{{\mathbf{v}}_{-x_i}} \sim \calD_{x_i|v^\prime_i}} \E_{{\mathbf{b}} \sim {\beta}({v_i', {\mathbf{v}}_{-x_i}})}\left[u_{x_i}\left(\mathbf{b};v_i'\right)\right] - 30\eps' \cdot 4 .
\]
As $\beta$ is an $\eps'$-approximate MBNE, it follows that 
\begin{align*}
     U_i(\sigma) &\geq W_i\E_{{{\mathbf{v}}_{-{x_i}}} \sim \calD_{x_i|v^\prime_i}} \E_{\mathbf{b} \sim {\beta}(v_i',{{\mathbf{v}}_{-x_i}})}\left[u_{x_i}\left(\mathbf{b};v_i'\right)\right] -  120\eps' W_i \\[0.5ex] &\geq W_i\E_{{{\mathbf{v}}_{-x_i}} \sim \calD_{x_i|v^\prime_i}} \E_{\mathbf{b}_{-x_i} \sim {\beta_{-x_i}}({{\mathbf{v}}_{-x_i}})}\left[u_{x_i}\left(b^* + 1,\mathbf{b}_{-x_i};v_i'\right)\right]- 121 \eps' W_i \\[0.3ex]
     &\geq U_i(b^\ast,\sigma_{-i}) -  121\eps' W_i.
\end{align*}
Using the choice of $\eps'=\eta/121$, this finishes the proof that $\sigma$ is an $\eta$-NE of $H$.

Finally, by dividing bids and values {in} $A_H$ by $5$, as described in \Cref{sec:defDFPA}, we have that there exists some constant $\eps$ such that computing an $\eps$-approximate MBNE of a DFPA is \PPAD-hard.
In particular, we can take $\eps=\eta/605$  and our hard DFPA instances have $B=\{0,1/5,2/5\}$.

\begin{proof}[Proof of \Cref{lem:small_0}]

    Consider a bidder $x_i$. For ease of notation we write $z_i:=\beta_{x_i}(v_i')(0)$ to denote the probability the bidder bids $0$ when they have value $v_i'$.

    We first consider the case $5 \geq v_i' \geq  3$ (meaning $v_i > 1/2$). Then, let $\beta^\ast$ be the bidding profile where $x_i$ bids $2$ instead of $0$ when they have value $v_i'$ (and remaining bidders are unchanged):  $$\beta^*_{x_i}(v^\prime_i)(2)=\beta_{x_i}(v^\prime_i)(2)+\beta_{x_i}(v^\prime_i)(0) \quad \text{ and }\quad \beta^*_{x_i}(v^\prime_i)(0)=0.$$
    
    We observe the following: 
    \begin{flushleft}\begin{itemize}
        \item When bidding $0$, $x_i$ will gain utility at most $\frac{5}{50}$ since the $50$ dummy bidders always bid $0$.
        \item  Conditioned on $x_i$ having value $v_i' \neq 0$, there are at most $3$ other bidders who have {nonzero} value and could submit a bid of $2$. Hence the utility of submitting bid $2$ when $x_i$ has value $v_i'$ is at least $(3-2)/4$. 
    \end{itemize}\end{flushleft}
Hence we have
    $$U_{x_i}(\beta^*_{x_i}, \beta_{-x_i}; v'_i) - U_{x_i}(\beta_{x_i}, \beta_{-x_i}; v'_i) \geq z_i \cdot (1/4-5/50) \geq z_i/10,$$

\noindent
but since $\beta$ is an $\eps$-MBNE, we must have $z_i \leq 10\eps$.\footnote{If we had a mixed strategy for bidder $x_i$ that performs better than $\beta_{x_i}$, then by averaging there would be some bid $b^*$ that performs better than $\beta_{x_i}$. } 

Now, consider a bidder $x_i$ with value $v_i'=3/2$ in $A_H$. Exactly half of the mass of $\calD_{x_i\, | 3/2}$ is on the point where $x_i$ is the sole bidder with nonzero value. Let $\beta^\ast$ be the bidding profile where bidder $x_i$ bids $1$ instead of $0$ when they have value
$v_i'$ (and remaining bidders are unchanged). That is $$\beta^*_{x_i}(v^\prime_i)(1)=\beta_{x_i}(v^\prime_i)(1)+\beta_{x_i}(v^\prime_i)(0) \quad \text{ and }\quad \beta^*_{x_i}(v^\prime_i)(0)=0.$$

Then as above we can conclude that $$
U_{x_i}(\beta^*_{x_i},\beta_{-x_i};v_i')
-
U_{x_i}(\beta_{x_i},\beta_{-x_i};v_i') \geq z_i \cdot \left(\frac{1}{2}\left(\frac{3}{2}-1\right) -\left(\frac{3}{2} \cdot \frac{1}{50}\right)\right) \geq z_i/10,$$ and thus $z_i \leq 10\eps$. This finishes the proof of the lemma.
\end{proof}

\subsection{Membership in $\PPAD$}

\begin{definition}\label{def:DFPA_decomposition}
    Given a DFPA $A=(\calN, B, \calV, \calD)${,} {we} define the hypergraph decomposition of $A$ as the HDFPA $H_A$ defined as $H_A:=(\calN', B, v',E)$ where:
\begin{itemize}
    \item[(a)] The set of bids will also be $B$.
    \item[(b)] For every bidder $i$ and every value $v_i \in V_i$ we create a bidder $x_{i,v_i}$, and we give  $x_{i,v_i}$ value $v_i$. 
    \item[(c)] For every  $v^{\ell} \in \supp(\calD)$, we create the hyperedge $e^{\ell}=(S^\ell,w_{e^\ell})$ where $$S^{\ell}=\left\{x_{1, v_1^{\ell}},\ldots,x_{n,v_n^{\ell}}\right\}\quad\text{and}\quad w_{e^{\ell}}=\calD(v^{\ell}).$$
\end{itemize}
\end{definition}

We have a bijection between the profiles that {are} MBNE of $A$ and profiles that are NE of $H_A$.  
\begin{lemma}\label{thm:from_HA_to_A}
    Consider a DFPA $A$ and {its} corresponding HDFPA decomposition $H_A$. There exists a function $f$ which maps {a} strategy profile of $H_A$ to a strategy profile of $A$:
    \begin{flushleft}\begin{enumerate}
        \item $f$ is a bijection. And $f$ and $f^{-1}$ can be computed in polynomial time.
        \item $\sigma$ is an $\eps$-NE of $H_A$ if and only if $f(\sigma)$ is an $\eps$-approximate MBNE of $A$.
        \item Furthermore, if $\sigma$ is a pure strategy profile for $H_A$, $f(\sigma)$ is pure. And if $\beta$ is a pure strategy profile for $A$ then ${f^{-1}}(\beta)$ is pure.  
    \end{enumerate}\end{flushleft}

\end{lemma}

We delay the proof of \Cref{thm:from_HA_to_A} to the end of this subsection. With the above lemma, it's easy to show that computing an $\eps$-approximate MBNE of a DFPA is in $\PPAD$.

\begin{lemma}
    The problem of computing an $\eps$-approximate MBNE of {a} DFPA, where $\eps>0$ is a rational number given in binary in the input{,} is in \PPAD{.} 
\end{lemma}
\begin{proof}
Let $A$ be a DFPA. It is clear that we can build $H_A$ from $A$ in time polynomial in the length of the representation of $A$. By \Cref{lem:HDFPA_in_PPAD} we know that computing  an $\eps$-NE $\sigma$ of $H_A$ is in $\PPAD$. By \Cref{thm:from_HA_to_A} given $\sigma$ we can compute in polynomial time an $\eps$-approximate MBNE of $A$. 
\end{proof}

For a bidder $y$ of $H_A$, a bid $b\in B$, a hyperedge $e$ and a profile $\sigma_{-y}$ of the other bidders, write $$u_y^e(b,\sigma_{-y}):=\E_{{\mathbf{b}}_{-y}\sim\sigma_{-y}}\left[u_y^e(b,{\mathbf{b}}_{-y})\right].$$

\begin{proof}[Proof of \Cref{thm:from_HA_to_A}]
    Let $A=(\calN, B, \calV, \calD)$ be a DFPA. Let $H_A$ be the {HDFPA} decomposition of~$A$.  
Let $\sigma$ be a strategy profile for $H_A$. For clarity, we use $U^{(H_A)}$ when talking about utilities related to  $H_A$ and $U^{(A)}$ to talk about utilities from $A$.

We define the map $f$ as the function which maps $\sigma$ to the strategy profile $\beta$ such that:
$$\beta_i(v):=\sigma_{x_{i,v}}{.}$$ 

That is, the strategy of bidder $i$ when they have value $v$ is the strategy of bidder $x_{i,v}$ in the HDFPA. It's easy to see this is a bijection and furthermore this function and its inverse are {polynomially} computable. 
It's also easy to see that $(3)$ holds, as $\beta$ is a pure profile iff $\sigma$ is pure.

It remains to prove that $\beta=f(\sigma)$ is an $\eps$-approximate MBNE of $A$ if and only if $\sigma$ is an $\eps$-NE of $H_A$. First assume $\sigma$ {is} an $\eps$-NE of $H_A$; we will show that for any bidder $i$, $v_i \in V_i$ and  $b^*\in B$:
\begin{equation}\label{eq:equlibrium_Auction}
U^{(A)}_{i}(\beta_i,\beta_{-i}; v_i) \geq U^{(A)}_{i}(b^*,\beta_{-i};v_i) -\eps. 
\end{equation}

So, fix a bidder $i$ of $A$ and a value $v_i \in V_i$ and a bid $b^*$. For ease of notation we let $y:=x_{i,v_i}$. We first observe that in $H_A$ we have:

$$W_{y}=\sum_{ \substack{v^{\ell} \in \supp(\calD),\\ v^{\ell}_{i}=v_i} }\calD(v^{\ell})=\calD_i(v_i){.}$$
So, since the strategy profile $\sigma$ {achieves} an $\eps$-NE of $H_A$, we know that
\begin{equation}\label{eq:beta_NE}
    U_{y}^{(H_A)}(\sigma) \geq U_y^{(H_A)}(b^*, \sigma_{-y}) -\eps W_y = U_y^{(H_A)}(b^*, \sigma_{-y}) -\eps \calD_i(v_i).
\end{equation}
We also observe that
\begin{align*}
    U^{(A)}_{i}(\beta_i,\beta_{-i}; v_i)&=\E_{{\mathbf{b}} \sim \beta_i(v_i)} \left[U_i^{(A)}({\mathbf{b}}, \beta_{-i}; v_i) \right] 
    =\E_{{\mathbf{b}} \sim \sigma_{y}} \left[U_i^{(A)}({\mathbf{b}}, \beta_{-i}; v_i)\right].
\end{align*}
{Here} we used $\beta_i(v_i)=\sigma_{y}$ by definition. Now for a fixed bid $b$ we have:
\begin{align*}
    U_{y}^{(H_A)}(b, \sigma_{-y})&=\sum_{\substack{e^{\ell}=(S^\ell,w_{e^\ell}) \in E,\\y \in S^\ell}} w_{e^{\ell}} \cdot u_{y}^{e^{\ell}}(b,\sigma_{-y}) =\sum_{ v^{\ell}\,:\, v^{\ell}_{i}=v_i} \calD\left( v^{\ell}\right)  u_{y}^{e^{\ell}}(b,\sigma_{-y}){.}
\end{align*}
Recall that by definition, we have: $\calD\left(v^{\ell}\right)=\calD_{i|v_i}\left(v^{\ell}_{-i}\right) \calD_{i}\left(v_i\right)$. Hence, 
\begin{align*}
       U_{y}^{(H_A)}(b, \sigma_{-y})&=\calD_i(v_i) \sum_{v^{\ell}\,:\, v^{\ell}_{i}=v_i} \calD_{i|v_i}(v^{\ell}_{-i})  u_{y}^{e^\ell}(b,\sigma_{-y}) =\calD_i(v_i) \E_{ v^{\ell}_{-i} \sim \calD_{i|v_i}}  u_{y}^{e^\ell}(b,\sigma_{-y}){.}
\end{align*}
\noindent
Observe that for any edge $e^\ell$ containing $y$, the utility $u_y^{e^\ell}(b,\sigma_{-y})$ only depends on the bids submitted by bidders $x_{j,v^\ell_j}$ in $e^\ell$ ($j \neq i$), and in particular $x_{j,v^\ell_j}$ submits a bid ${\mathbf{b}}_j \sim \sigma_{x_{j,v^{\ell}_j}}$. So, using $\beta_j(v^{\ell}_j)=\sigma_{x_{j,v^{\ell}_j}}$ we can see that $$u_y^{e^\ell}(b,\sigma_{-y})=\E_{\mathbf{b}_{-i} \sim \beta_{-i}( v^{\ell}_{-i})}u_i(b,\mathbf{b}_{-i};v_i).$$
\noindent
Putting it all together, {we} have:
\begin{align*}
    U_{y}^{(H_A)}(b, \sigma_{-y})&=\calD_i(v_i) \E_{ \mathbf{v}_{-i} \sim \calD_{i|v_i}}  \E_{\mathbf{b}_{-i} \sim \beta_{-i}( \mathbf{v}_{-i})}u_i(b,\mathbf{b}_{-i};v_i)= \calD_i(v_i) U^{(A)}_i(b,\beta_{-i}; v_i){.}
\end{align*}

 We can thus conclude that for any bid $b$ we have $$U^{(A)}_i(b,\beta_{-i}; v_i)=U_{y}^{(H_A)}(b, \sigma_{-y})/\calD_i(v_i).$$ In particular, we have
 $$U_{i}^{(A)}(\beta_i, \beta_{-i}; v_i)=U_{y}^{(H_A)}(\sigma)/\calD_i(v_i) \quad \text{ and } \quad U_{i}^{(A)}(b^*, \beta_{-i}; v_i)=U_{y}^{(H_A)}(b^*,\sigma_{-y})/\calD_i(v_i){.} $$
 {Hence, we} can conclude that \Cref{eq:equlibrium_Auction} {is} equivalent to \Cref{eq:beta_NE}. Thus, we can conclude that $\beta$ is an $\eps$-approximate MBNE of $A$. Since \Cref{eq:equlibrium_Auction,eq:beta_NE} are equivalent, we also have that if $\beta$ is an $\eps$-approximate MBNE of $A$, then $f^{-1}(\beta)$ must be an $\eps$-NE of $H_A$. 
\end{proof}

\noindent\textbf{AI Methodology.} %
{The new HDFPA model arose in
discussions among the authors, without AI assistance.
With input-output relations specified by the authors, 
ChatGPT~5.6 Sol was used in the design of the three gadgets in the proof of \Cref{thm:HDFPA_ppad_hard}, which was particularly helpful in setting up weights of the auctions. We also used Claude Opus~5.5 for
literature search and writing suggestions, and ChatGPT~6 Astra to proofread the paper. All proofs were written and verified by the authors, who take full responsibility for the content of the paper.
}

\bibliographystyle{alpha}
\bibliography{the}{}

\appendix
\section{NP-completeness of Pure Equilibria in HDFPA}\label{sec:npc}

In this section we prove the following theorem:
\begin{theorem}\label{thm:NP_hard_HDFPA}
     There exists $\eps$ of size inverse-polynomial in the length of the problem description such that, for all $\eps'<\eps$, the problem of deciding the existence of a pure
$\eps'$-PNE of a HDFPA is
NP-complete.
\end{theorem}

We recall the following theorem on the hardness of deciding the existence of a pure equilibrium in a DFPA from \cite{filosratsikas2025equilibriumcomputationfirstpriceauctions}. 
\begin{theorem}\label{thm:DFPA_NP}
     There exists $\eta$ of size inverse-polynomial in the length of the problem description such that, for all $\eta'<\eta$, the problem of deciding the existence of a pure
$\eta'$-PBNE of a DFPA with correlated priors is
NP-hard.
\end{theorem}
\begin{proof}[Proof of \Cref{thm:NP_hard_HDFPA}]

     We first show the problem is NP-hard by showing that checking the existence of an $\eps$-PBNE of {a} DFPA is polytime reducible to checking the existence of {an} $\eps$-PNE of a HDFPA. Let $A=(\calN, B, \calV, \calD)$ be a DFPA and let $H_A$ be {its} hypergraph decomposition (see \Cref{def:DFPA_decomposition}).  The description of $H_A$ can be computed in polynomial time in the description of $A$.

     Consider the function $f$ guaranteed to {exist} by \Cref{thm:from_HA_to_A}. From its existence we can conclude that for any $\eps\geq 0$ we have that $A$ has an $\eps$-PBNE {if and only if} $H_A$ has {an} $\eps$-PNE. Indeed, $f$ is a bijection between profiles that are pure $\eps$-PNE of $H_A$ and profiles that are pure $\eps$-PBNE of $A$.

     Now we observe that the description {sizes} of $A$ and $H_A$ are polynomially related to each other. Hence, the proof follows immediately from \Cref{thm:DFPA_NP} since the $\eta$ guaranteed to {exist} by the theorem is also of size inverse polynomial in the description size of $H_A$. 

     \medskip
     Membership in NP follows easily. A certificate {is} a pure profile $\sigma$ with no overbidding. Since {each} bidder follows a pure strategy in $\sigma$, the utility $U_i(b,\sigma_{-i})$ can be computed in polynomial {time} for every bid $b$ and bidder $i$ in the size of $H$. Hence, we can check in polynomial time if $\sigma$ is {an} $\eps$-PNE of $H$ or not. 
\end{proof}

\end{document}